\RequirePackage{fix-cm}
\documentclass[smallextended,envcountsect,envcountsame]{svjour3}       % onecolumn (second format)

\usepackage{amsmath} 
\usepackage{amssymb}
\usepackage{tikz}
\usepackage[noend]{algpseudocode}
\usepackage{algorithm}
\usepackage{graphicx,parskip,subcaption}

\algnewcommand\algorithmicforeach{\textbf{for each}}
\algdef{S}[FOR]{ForEach}[1]{\algorithmicforeach\ #1\ \algorithmicdo}

\usepackage{thmtools}
\usepackage{thm-restate}
\usepackage{hyperref}

\newtheorem{lem}[theorem]{Lemma}
\newtheorem{cor}[theorem]{Corollary}
\newtheorem{prop}[theorem]{Proposition}
\newtheorem{proc}[theorem]{Procedure}

\smartqed

\usepackage{microtype}%if unwanted, comment out or use option "draft"
\usepackage{todonotes}

\begin{document}
	
	\title{Efficient Isomorphism for $S_d$-graphs and $T$-graphs
		\thanks{This paper is the full extended version of the conference paper that appeared at MFCS 2020.
			It contains the detailed algorithms and full proofs, and few additional small results.}
		\thanks{Supported by research grant GA\v CR 20-04567S of the Czech Science Foundation.}
	}
	
	\author{Deniz A\u{g}ao\u{g}lu \c{C}a\u{g}{\i}r{\i}c{\i} \and Petr Hlin\v en\'y }
	
	\institute{D. A\u{g}ao\u{g}lu \c{C}a\u{g}{\i}r{\i}c{\i} \at
		Masaryk University, Brno, Czech Republic \\
		\email{agaoglu@mail.muni.cz}           %  \\
		\and
		P. Hlin\v{e}n\'y$^*$ (corresponding author) \at
		Masaryk University, Brno, Czech Republic \\
		\email{hlineny@fi.muni.cz}$^*$ 
	}

	\maketitle
	
	\begin{abstract}
		An $H$-graph is one representable as the intersection graph of connected subgraphs of a suitable
		subdivision of a fixed graph $H$, introduced by Bir\'{o}, Hujter and Tuza (1992).  
		An $H$-graph is proper if the representing subgraphs of $H$ can be chosen
		incomparable by the inclusion.
		In this paper, we focus on the isomorphism problem for $S_d$-graphs and
		$T$-graphs, where $S_d$ is the star with $d$ rays and $T$ is an arbitrary
		fixed tree.  
		
		Answering an open problem of Chaplick, T\"{o}pfer, Voborn\'{\i}k and Zeman (2016),
		we provide an \textbf{FPT}-time algorithm for testing isomorphism
		and computing the automorphism group 
		of $S_d$-graphs when parameterized by~$d$,
		which involves the classical group-computing machinery by Furst, Hopcroft, and Luks (1980).
		We also show that the isomorphism problem of $S_d$-graphs is at
		least as hard as the isomorphism problem of posets of bounded width,
		for which no efficient combinatorial-only algorithm is known to date.
		Then we extend our approach to an \textbf{XP}-time algorithm 
		for isomorphism of $T$-graphs when parameterized by the size of $T$.
		Lastly, we contribute a simple \textbf{FPT}-time combinatorial algorithm for
		isomorphism testing in the special case of proper $S_d$- and $T$-graphs.

		\keywords{intersection graph; isomorphism testing; chordal graph; $H$-graph; parameterized complexity}%mandatory
	\end{abstract}

	\section{Introduction}\label{introduction}
	
	A \emph{graph} is a pair $G=(V,E)$ where $V=V(G)$ is the \emph{finite} vertex set 
	and $E=E(G)$ is the edge set -- a set of unordered pairs of vertices.
	A \emph{subdivision} of an edge $\{u,v\}$ of a graph $G$ is the operation of replacing
	$\{u,v\}$ with a new vertex $x$ and two new edges $\{u,x\}$ and~$\{x,v\}$.
	Two graphs $G_1$ and $G_2$ are called \emph{isomorphic} and denoted by $G_1\simeq G_2$,  
	if~there exists a bijection $f: V(G_1) \rightarrow V(G_2)$, called an 
	{\em isomorphism}, such that $\{ u,v \} \in E(G_1)$ if and only if $\{ f(u),f(v)\} \in E(G_2)$ for all $\{ u,v \} \subseteq V(G_1)$.
	
	The \emph{graph isomorphism problem} is to determine whether the two given graphs are isomorphic. 
	It is in a sense a quite special problem in computer science;
	on one hand, under some widely-believed complexity-theoretic assumptions, it can be shown that
	graph isomorphism is not an NP-hard problem, while on the other hand,
	a polynomial-time algorithm for graph isomorphism is still elusive (and not
	everybody expects existence of such algorithm).
	It has actually defined its own complexity class \emph{GI} of the problems
	which are reducible in polynomial time to graph isomorphism.
	The current state of the art is a quasi-polynomial algorithm of Babai~\cite{DBLP:conf/stoc/Babai16}.
	Nevertheless, the problem has been shown to be solvable efficiently for various
	natural graph classes such as trees, planar and permutation graphs
	\cite{AHU,planarLinear,DBLP:journals/networks/Colbourn81}
	and for parameterized classes such as those listed below. 
	
	A (decision) problem with a parameter $k$ belongs to the class \textbf{FPT} if it can be
	solved in time $f(k)\cdot n^{\mathcal{O}(1)}$ where $f$ is a
	computable~function and $n$ is the size of the input instance.  Similarly, a
	decision problem with a parameter $k$ belongs to \textbf{XP} if it can be
	solved in time $f(k)\cdot n^{g(k)}$ where $f$ and $g$ are two
	computable~functions and $n$ is the size of the input.  Even though
	the graph isomorphism problem has not been studied as intensively as other
	``classical'' graph problems in the parameterized setting, 
	some of the well-known parameterizations yielding to
	\textbf{FPT}- and \textbf{XP}-time algorithms are the maximum
	degree~\cite{DBLP:journals/jcss/Luks82},
	eigenvalue multiplicity~\cite{eigenmul},
	genus~\cite{genus,DBLP:journals/corr/Kawarabayashi15a,DBLP:conf/esa/Neuen21},
	tree-depth~\cite{treedepth} and tree-width~\cite{treewidth}.

	Now, let us briefly introduce the graph classes which are the
	subject of our research.  The \emph{intersection graph} $G$ of a finite
	collection of sets $\{S_1, \dots, S_n\}$ is a simple undirected graph in which each set
	$S_i$ is associated with a vertex $v_i \in V(G)$ and each pair $v_i,v_j$ of
	vertices is joined by an edge if and only if the corresponding sets have a
	non-empty intersection, i.e.  $\{v_i,v_j\} \in E(G) \iff S_i \cap S_j \neq \emptyset$.
	The relevant special collections of sets are described below.
	
	A graph is \emph{chordal} if every induced cycle of length more than three
	has a chord.  This can be defined as the intersection graph of subtrees of
	some suitable tree~\cite{chordalityInters}.  Chordal graphs have linearly many
	maximal cliques which can be listed in polynomial time
	\cite{recogChordaLinear}.  Deciding the isomorphism of chordal graphs is a
	{\em GI-complete} problem \cite{isoChordalGIComp}.  This means that testing
	whether two chordal graphs are isomorphic is polynomial-time equivalent to
	the graph isomorphism problem in the general case.
	
	A graph $G$ is an \emph{interval graph} if it is the intersection graph for a set of intervals on the real line, 
	and interval graphs form a subclass of chordal graphs.
	The isomorphism problem for interval graphs can be solved in linear time~\cite{recogIntervalLinear}.
	
	\emph{Split graphs} are chordal graphs whose vertex set can be partitioned
	into a clique and an independent set.  They present a special case of
	intersection graphs of substars of a suitable subdivided star, and the
	isomorphism problem for split graphs is also {\em GI-complete}
	\cite{isoSplitGIComp}.
	
	For a fixed graph $H$, an {\em$H$-graph} is the intersection graph of connected 
	subgraphs of a suitable subdivision $H'$ of the graph~$H$~\cite{biro}. 
	Such an intersection representation is also called an $H$-representation.
	They generalize all mentioned intersection graphs as follows. Interval graphs are $K_2$-graphs, chordal graphs are the union of $T$-graphs where $T$ ranges over all trees, and split graphs are contained in the union of $S_d$-graphs where $d$ ranges over all positive integers.
	Every $S_d$-graph is chordal, but not every $S_d$-graph is a split graph.
	Various optimization problems such as maximum clique and minimum
	dominating set on $H$-graphs (for particular graphs~$H$) have been shown to be solvable in
	polynomial and/or \textbf{FPT}-time \cite{zemanWG,zeman2,DBLP:conf/esa/FominGR18}.
	
	If a graph $G$ has an $H$-representation, i.e., an intersection
	representation by subgraphs of a subdivision $H'$ of $H$,
	such that these subgraphs of $H'$ are pairwise incomparable by inclusion
	(no one is a subgraph of another), then we speak about a {\em proper
		representation} and $G$ is called a {\em proper $H$-graph}.
	Obviously, a proper $H$-graph is also an $H$-graph, but the converse is far
	from being true.
	For instance, the class of proper $K_2$-graphs coincides with the class of
	unit interval graphs.

	\subsection{Organization of the paper}
	\begin{itemize}
		\item In Section~\ref{introSD}, we give the definitions and
		basic properties of $S_d$- and $T$-graphs in closer detail.
		
		\item In Section~\ref{easysection}, we consider $S_d$-graphs with bounded
		maximum clique size at most~$p$ and give, as a warm-up exercise,
		a simple combinatorial \textbf{FPT}-time isomorphism algorithm parameterized only by $p$ (Theorem~\ref{theo:SDSMallTheo}). 
		
		\item In Section~\ref{reductionsection}, we prove that the $S_d$-graph isomorphism problem includes isomorphism testing of posets of width~$d$ (Theorem~\ref{theo:red2})
		which can be solved using the group-based approach by Furst, Hopcroft and Luks~\cite{furst} via Babai~\cite{babai-bdcm}
		(but not known to have a combinatorial algorithm).
		
		\item In Section~\ref{hardsection}, as our main result,
		we combine the case of posets of bounded width with a specific adaptation of the general group-compu\-ting approach by Furst, Hopcroft and Luks~\cite{furst} to obtain an \textbf{FPT}-time isomorphism algorithm for $S_d$-graphs (Theorem~\ref{thm:MAIN}).
		
		\item In Section~\ref{hardestsection}, we extend our result on $S_d$-graph isomorphism to $T$-graph isomorphism, and obtain an \textbf{XP}-time isomorphism algorithm for $T$-graphs.
		This algorithm actually applies to all chordal graphs parameterized by their leafage.
		
		\item  In Section~\ref{proper}, we focus on the isomorphism problem for proper $S_d$- and $T$-graphs, and show that their isomorphism can be tested in \textbf{FPT}-time by a combinatorial algorithm.
		
	\end{itemize}
	We remark that all our algorithms expect only (abstract) graphs on
	the input, i.e., they do not require an intersection representation of the graph to be given.

	\section{ $S_d$-graphs and $T$-graphs}\label{introSD}
	
	We first introduce necessary details about representations of $S_d$- and $T$-graphs. 
	For basic poset terms related to $S_d$- and $T$-representations, we refer the
	readers to~\cite{posetbook}.
	We just recall that the {\em width} of a poset is the maximum size of its
	antichain, and this number is equal to the minimum number of chains covering
	all poset elements.
	
	An {\em$S_d$-graph} $G$ is the intersection graph of connected 
	substars of a suitable subdivision $S'$ of the star $S_d$. 
	In such a representation, every ray of $S'$ defines an induced interval subgraph of $G$, 
	and the central node of~$S'$ defines a clique $C$ of $G$.
	We may always straightforwardly modify the representation such that $C$ is a
	maximal (by inclusion) clique of~$G$.
	For further reference, we call a maximal clique $C$ of $G$ a {\em central clique}
	of $G$ if there is an $S_d$-representation of $G$ in which the central node defines~$C$.
	
	Given a graph $G$ and a maximal clique $C$ of $G$, let $\mathcal{X}=\{X_1,X_2,\ldots,X_c\}$ denote
	the set of connected components of $G-C$.  For each connected component $X_i \in
	\mathcal{X}$, let $N_C{(X_i)}$ be the set of neighbors of $X_i$ in $C$ which
	is called the {attachment} of~$X_i$.  A~connected component together
	with its attachment edges is called a \emph{bridge} of~$C$ in~$G$.
	Chaplick et al.~\cite{zemanWG} obtained a characterization of
	$S_d$-graphs yielding to a polynomial time recognition algorithm for
	arbitrary~$d$.  This useful characterization is based
	on the following partial order $P$ on the connected components of $G-C$
	which we will call the {\em central poset} of $G$ on the clique~$C$.
	
	Since each $G[C \cup X_i]$ induces an interval subgraph, the neighborhoods
	of the vertices of $X_i$ in $C$ form a chain by inclusion.  
	The {\em upper attachment} of~$X_i$, denoted by
	${N_C}^U{(X_i)}$, is the maximum neighborhood in $C$ among the vertices of~$X_i$,
	and it is $N_C{(X_i)} = {N_C}^U{(X_i)}$.  Analogously, the {\em lower attachment},
	denoted by ${N_C}^L{(X_i)}$, is the minimum neighborhood in $C$ among the vertices of $X_i$.  
	Note that the upper and lower attachments are well-defined since they are
	unique for each $X_i$ in~$G$.
	After the {\em attachment} (i.e., the lower and upper one)
	of each connected component of $G-C$ is determined, 
	the partial order $P$ is constructed by comparing the
	attachments of each pair of connected components.  A pair $(X_i,X_j)$ of
	components is {\em comparable in $P$}, denoted by $X_i \preceq_{P} X_j$, 
	if ${N_C}^U{(X_i)} \subseteq {N_C}^L{(X_j)}$ holds, and incomparable
	if neither of $X_i \preceq_{P} X_j$, $X_j \preceq_{P} X_i$ holds.
	We will also speak about {\em comparable/incomparable attachments}
	(in $C$) of the components.
	
	Naturally, for every chain $X_1 \preceq_{P}\dots \preceq_{P}
	X_k$, the induced subgraph $G[C \cup X_k \cup \dots \cup X_1]$ is also an
	interval graph.  In addition, distinct connected components $X_i$ and $X_j$
	of $G-C$ are called \emph{equivalent} and treated as one bridge of $C$ when
	${N_C}^U(X_i) = {N_C}^L(X_i) = {N_C}^U(X_j) = {N_C}^L(X_j)$ holds.
	
	The mentioned characterization simply reads:
	
	\begin{prop}[Chaplick et al.~\cite{zemanWG}]\label{prop:Sdcharacteriz}
		A graph $G$ is an $S_d$-graph if and only if there exists a maximal clique $C$ of $G$ such
		that, for $\mathcal{X}$ and $P$ as above,
		\begin{itemize}
			\item for all $X_i \in \mathcal{X}$, the induced subgraph $G[C \cup X_i]$ is an interval graph, and
			\item the central poset $P$ (of $G$ on~$C$) can be covered by at most $d$ chains.
		\end{itemize}
	\end{prop}

	\begin{figure}[tbp]
		\centering
		\captionsetup[subfigure]{position=b,justification=centering}
		\begin{subfigure}[t]{0.48\linewidth}
			\centering
			\begin{tikzpicture}[scale=0.75]
				
				\node (a) at (-0.2,0.4) {3};
				\draw[rotate around={90:(-0.2,0.4)}, color=orange, fill=orange, ultra thick, opacity=0.3] (-0.2,0.4) ellipse (0.25cm and 0.25cm); 
				
				\node (b) at (-0.2,2) {1};
				\draw[rotate around={90:(-0.2,2)}, color=orange, fill=orange, ultra thick, opacity=0.3] (-0.2,2) ellipse (0.25cm and 0.25cm); 
				
				\node (c) at (-1.2,1.2) {2};
				\draw[rotate around={90:(-1.2,1.2)}, color=orange, fill=orange, ultra thick, opacity=0.3] (-1.2,1.2) ellipse (0.25cm and 0.25cm);
				
				\node (d) at (0.8,1.2) {4};
				\draw[rotate around={90:(0.8,1.2)}, color=orange, fill=orange, ultra thick, opacity=0.3] (0.8,1.2) ellipse (0.25cm and 0.25cm);
				
				\draw[rotate around={90:(-2.8,1.2)}, color=magenta, fill=magenta, ultra thick, opacity=0.3] (-2.8,1.2) ellipse (0.35cm and 0.35cm); %123
				\draw[rotate around={90:(2.8,1.2)}, color=blue, fill=blue, ultra thick, opacity=0.3] (2.8,1.2) ellipse (0.7cm and 0.35cm); %13
				\draw[rotate around={90:(-1.2,3)}, color=green, fill=green, ultra thick, opacity=0.3] (-1.2,3) ellipse (0.35cm and 0.7cm);%1
				\draw[rotate around={90:(0.8,3)}, color=teal, fill=teal, ultra thick, opacity=0.3] (0.8,3) ellipse (0.35cm and 0.35cm);%14
				\draw[rotate around={90:(1.6,0)}, color=red, fill=red, ultra thick, opacity=0.3] (1.6,0) ellipse (0.35cm and 0.7cm);%4
				\draw[rotate around={90:(-1.8,0)}, color=cyan, fill=cyan, ultra thick, opacity=0.3] (-1.8,0) ellipse (0.35cm and 0.35cm);%23
				
				\node [label=left:{$X_1~$}] at (-2.8,1.2) [circle,draw, fill=black, opacity=1, color=black, inner sep=0.8mm] (e) {};
				\node at (2.8,0.8) [circle,draw, fill=black, opacity=1, color=black, inner sep=0.8mm] (f) {};
				\node [label=right:{$~X_2$}] at (2.8,1.6) [circle,draw, fill=black, opacity=1, color=black, inner sep=0.8mm] (g) {};
				\node at (-0.8,3) [circle,draw, fill=black, opacity=1, color=black, inner sep=0.8mm] (h) {};
				\node [label=left:{$X_3~$}] at (-1.6,3) [circle,draw, fill=black, opacity=1, color=black, inner sep=0.8mm] (i) {};
				\node [label=right:{$~X_4$}] at (0.8,3) [circle,draw, fill=black, opacity=1, color=black, inner sep=0.8mm] (j) {};
				\node [label=left:{$X_5~$}] at (1.2,0) [circle,draw, fill=black, opacity=1, color=black, inner sep=0.8mm] (k) {};
				\node at (2,0) [circle,draw, fill=black, opacity=1, color=black, inner sep=0.8mm] (l) {};
				\node [label=left:{$X_6~$}] at (-1.8,0) [circle,draw, fill=black, opacity=1, color=black, inner sep=0.8mm] (m) {};
				
				\draw (a) -- (b) node[midway, above]{};
				\draw (a) -- (c) node[midway, above]{};
				\draw (a) -- (d) node[midway, above]{};
				\draw (b) -- (c) node[midway, above]{};
				\draw (b) -- (d) node[midway, above]{};
				\draw (c) -- (d) node[midway, above]{};
				\draw (c) -- (e) node[midway, above]{};
				\draw (a) -- (e) node[midway, above]{};
				\draw (b) -- (e) node[midway, above]{};
				\draw (f) -- (g) node[midway, above]{};
				\draw (b) -- (f) node[midway, above]{};
				\draw (a) -- (g) node[midway, above]{};
				\draw (b) -- (g) node[midway, above]{};
				\draw (h) -- (i) node[midway, above]{};
				\draw (h) -- (b) node[midway, above]{};
				\draw (j) -- (b) node[midway, above]{};
				\draw (j) -- (a) node[midway, above]{};
				\draw (j) -- (d) node[midway, above]{};
				\draw (k) -- (d) node[midway, above]{};
				\draw (k) -- (l) node[midway, above]{};
				\draw (m) -- (a) node[midway, above]{};
				\draw (m) -- (c) node[midway, above]{};
				
				\node (11) at (-0.2,-0.8) {(a)};
				
			\end{tikzpicture}	
			\label{fig:UDVGa}
		\end{subfigure}
		~
		\begin{subfigure}[t]{0.48\linewidth}
			\centering
			\begin{tikzpicture}[scale=0.75]
				
				\draw[rotate around={90:(7,2.5)}, color=magenta, fill=magenta, ultra thick, opacity=0.3] (7,2.5) ellipse (0.35cm and 0.35cm); %123
				\draw[rotate around={90:(7.5,1.4)}, color=blue, fill=blue, ultra thick, opacity=0.3] (7.5,1.4) ellipse (0.35cm and 0.35cm); %13
				\draw[rotate around={90:(8,0.3)}, color=green, fill=green, ultra thick, opacity=0.3] (8,0.3) ellipse (0.35cm and 0.35cm); %1
				\draw[rotate around={90:(6,1.4)}, color=cyan, fill=cyan, ultra thick, opacity=0.3] (6,1.4) ellipse (0.35cm and 0.35cm); %23
				\draw[rotate around={90:(8.5,2.5)}, color=teal, fill=teal, ultra thick, opacity=0.3] (8.5,2.5) ellipse (0.35cm and 0.35cm); %14
				\draw[rotate around={90:(10,0.3)}, color=red, fill=red, ultra thick, opacity=0.3] (10,0.3) ellipse (0.35cm and 0.35cm); %4
				
				\node [label=left:{$\{1, 2, 3\}$}] (1) at (7,2.5) {$X_{1}$};
				\node [label=left:{$\{2, 3\}$}] (6) at (6,1.4) {$X_{6}$};
				\node [label=right:{$\{1, 3\}$}] (3) at (7.5,1.4) {$X_{2}$};
				\node [label=right:{$\{1\}$}] (2) at (8,0.3) {$X_{3}$};
				\node [label=right:{$\{1, 3, 4\}$}] (4) at (8.5,2.5) {$X_{4}$};
				\node [label=right:{$\{4\}$}] (5) at (10,0.3) {$X_{5}$};
				
				\draw[->] (6)--(1);
				\draw[->] (3)--(1);
				\draw[->] (2)--(3);	
				\draw[->] (3)--(4);		
				\draw[->] (5)--(4);
				
				\node (11) at (8,-0.8) {(b)};
				
			\end{tikzpicture}
			\label{fig:UDVGb}
		\end{subfigure}
		
		\begin{subfigure}[t]{0.44\linewidth}
			\centering
			\begin{tikzpicture}[scale=0.75]
				
				\draw[black, dashed, ultra thick] (4.75,2.75) -- (4.75,0.5) -- (1.75,-1);
				\draw[black, dashed, ultra thick] (4.75,0.5) -- (7.75,-1);
				
				\node [label=left:{\textcolor{magenta}{$X_1$}}] (111) at (3.65,0.85) { };
				\draw[magenta, line width=2.6] (3.45,0.85) -- (4,1.15); %X1
				
				\node [label=left:{\textcolor{blue}{$X_2$}}] (222) at (2.6,0.2) { };
				\draw[blue, line width=2.6] (2.45,0.2) -- (3,0.5); %X2
				
				\node [label=left:{\textcolor{green}{$X_3$}}] (333) at (2.15,-0.7) { };
				\draw[green, line width=2.6] (2,-0.7) -- (2.65,-0.35); %X3
				
				\node [label=right:{\textcolor{teal}{$X_4$}}] (444) at (5.8,0.25) { };
				\draw[teal, line width=2.6] (5.5,0.55) -- (6,0.3); %X4
				
				\node [label=right:{\textcolor{red}{$X_5$}}] (555) at (7.15,-0.55) { };
				\draw[red, line width=2.6] (6.85,-0.25) -- (7.35,-0.5); %X5
				
				\node [label=right:{\textcolor{cyan}{$X_6$}}] (666) at (5.05,2.55) { };
				\draw[cyan, line width=2.6] (5.2,1.95) -- (5.2,2.55); %X6
				
				\node [label=below:{1}] (1) at (2.775,-0.7) { };
				\draw[orange, ultra thick] (2.625,-0.825) -- (4.75,0.25) -- (6.25,-0.5);
				
				\node [label=left:{2}] (2) at (3.7,0.4) { };
				\draw[orange, ultra thick] (3.5,0.5) -- (5,1.25) -- (5,2.25);
				
				\node [label=left:{3}] (3) at (3.2,-0.1) { };
				\draw[orange, ultra thick] (3,0) -- (4.5,0.75) -- (4.5,2.25) -- (4.5,0.75) -- (6.25,-0.125);
				
				\node [label=below:{4}] (4) at (6.6,-0.875) { };
				\draw[orange, ultra thick] (4.75,0) -- (6.75,-1);
				
				\node (11) at (4.75,-1.5) {(c)};
				
			\end{tikzpicture}
			\label{fig:UDVGd}
		\end{subfigure}
		
		\caption{(a) An $S_d$-graph $G$ with its central clique $C = \{ 1, 2, 3, 4\}$ and the connected components of $G-C$. (b) The partial order $P$ on the connected components. (c) An $S_d$-representation of $G$ with $C$ in the center.}		
		\label{fig:SDGraph}
	\end{figure}
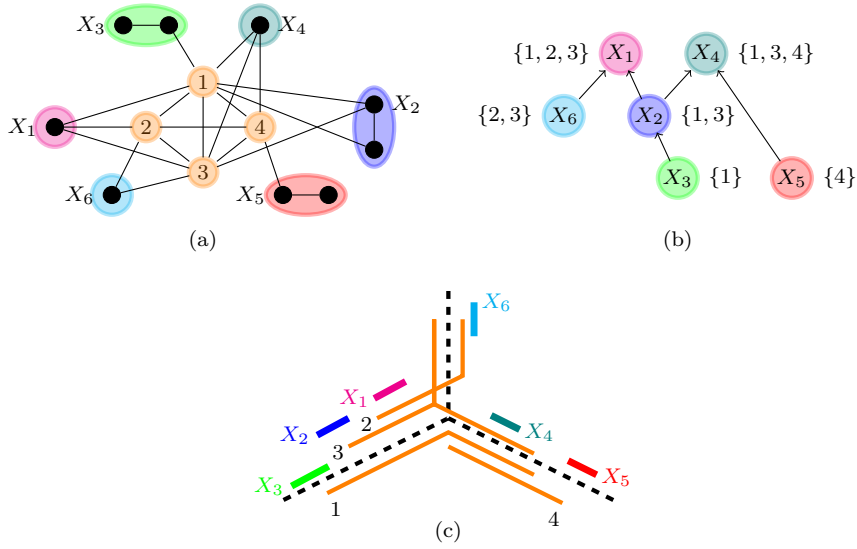
	
	In Figure~\ref{fig:SDGraph} (a), we demonstrate an $S_d$-graph with its
	central clique $C$ colored orange, and the connected components $X_1, \dots, X_6$ of $G-C$ colored differently. 
	In Figure~\ref{fig:SDGraph} (b), we see the partial order $P$ on the connected components of $G-C$ with respect to their attachments. The three tuples of components $(X_1,X_2,X_3)$, $(X_4,X_5)$, $(X_6)$ form a chain cover of $P$. Therefore, $G$ is an $S_d$-graph for $d=3$. In Figure~\ref{fig:SDGraph} (c), the corresponding $S_d$-representation of $G$ is given where the connected components are placed on the rays of a subdivision of $S_3$ according to this chain cover, and
	$C$ is placed in the center.
	
	If $G$ is not connected, then the characterization in
	Proposition~\ref{prop:Sdcharacteriz} says that all components of $G$ not
	containing the central clique $C$ are interval graphs,
	and their attachments are empty (hence they are treated as one bridge in our setting, and can be placed together at the end of any edge of~$S_d$).
	
	Note that the intersection representation of an $S_d$-graph is not uniquely determined by a central maximal clique. For instance, the tuples $(X_1,X_6)$, $(X_2,X_3)$, $(X_4,X_5)$ also form a chain cover of the partial order given in Figure~\ref{fig:SDGraph} (b), and it leads to another $S_d$-representation. In particular, there can be $d^{\Omega(n)}$ many distinct chain covers of a poset of width $d$ on $n$ elements when the depth of $P$ is not bounded. This means that the isomorphism problem of $S_d$-graphs cannot be simply solved by comparing the rays using the linear time isomorphism testing for interval graphs~\cite{recogIntervalLinear} for suitable pairs of maximal cliques.
	
	We actually have an easy observation:
	
	\begin{prop}[originally proved in \cite{isoChordalGIComp}]\label{prop:SdGIc}
		The isomorphism problem of proper $S_d$-graphs (and hence of all $S_d$-graphs) with $d$ on the input is GI-complete.
	\end{prop}
	\begin{proof}
		Let $G_1$ and $G_2$ be two arbitrary graphs on the same number of at least $4$ vertices.
		We construct~$G_i'$, $i=1,2$, as follows: subdivide every edge with a new
		vertex, and then make a clique on the original vertex set $V(G_i)$.
		Then, $G_1\simeq G_2$ if and only if $G_1'\simeq G_2'$; this is easy since
		any isomorphism of $G_1'$ and $G_2'$ must map $V(G_1)$ to $V(G_2)$ due to vertex degrees.
		Since each $G_i'$ is an $S_d$-graph for
		$d=|E(G_i)|$, with the central clique on $V(G_i)$,
		solving their isomorphism would solve also the isomorphism of $G_1$ and $G_2$.
		Furthermore, $G_i'$ is a proper $S_{d'}$-graph for $d'=|E(G_i)|+|V(G_i)|$, where
		the additional $|V(G_i)|$ rays of $S_{d'}$ are used to make the subgraphs
		representing the central clique pairwise incomparable.
		\qed\end{proof}

	More generally,
	a $T$-graph $G$ is the intersection graph of connected subtrees of a
	suitable subdivision $T'$ of a fixed tree $T$.  In such a representation, the
	path between any pair of branching nodes of $T'$ (i.e., nodes of degree at least~$3$) 
	defines an induced interval subgraph of $G$, and each branching node of $T'$ defines a clique
	of $G$ which can again be assumed maximal in $G$.
	$S_d$-graphs are $T$-graphs for a tree $T$ with only one branching node, and
	$T$-graphs are also chordal.  While $S_d$-graphs can be recognized in
	polynomial time for an arbitrary $d$ \cite{zemanWG}, the recognition problem
	for $T$-graphs is NP-complete when $T$ is on the input \cite{KLAVIK201585}.

	In \cite{zemanWG}, Chaplick et al. proved that a graph $G$ is a $T$-graph for some fixed tree $T$ if and only if there exists a set of maximal cliques $C_1,\dots,C_k$ of $G$ placed on the branching
	nodes $b_1, \dots, b_k$ of $T$ such that, for each $b_i$, the induced subgraph formed by the maximal clique $C_i$ placed on $b_i$
	and a selection of the connected components of $G-(C_1 \cup \dots \cup C_k)$ is an $S_d$-graph with $d=deg(b_i)$ and with additional restrictions detailed in~\cite{zemanWG}.
	Using this characterization, they showed that $T$-graphs can be recognized in \textbf{XP}-time with
	respect to $|V(T)|$ by checking all assignments of maximal cliques to the branching nodes of $T$ by brute-force.
	
	\begin{figure}[tbp]
		\centering
		\begin{subfigure}[t]{0.9\linewidth}
			\centering
			\begin{tikzpicture}[scale=0.75]
				
				\draw[rotate around={90:(-0.2,0.4)}, color=orange, fill=orange, ultra thick, opacity=0.3] (-0.2,0.4) ellipse (0.25cm and 0.25cm); 
				\node (a) at (-0.2,0.4) {3};
				\draw[rotate around={90:(-0.2,2)}, color=orange, fill=orange, ultra thick, opacity=0.3] (-0.2,2) ellipse (0.25cm and 0.25cm); 
				\node (b) at (-0.2,2) {1};	
				\draw[rotate around={90:(-1.2,1.2)}, color=orange, fill=orange, ultra thick, opacity=0.3] (-1.2,1.2) ellipse (0.25cm and 0.25cm);
				\node (c) at (-1.2,1.2) {2};
				\draw[rotate around={90:(0.8,1.2)}, color=orange, fill=orange, ultra thick, opacity=0.3] (0.8,1.2) ellipse (0.25cm and 0.25cm);
				\node (d) at (0.8,1.2) {4};
				
				\draw[rotate around={90:(-0.6,3)}, color=magenta, fill=magenta, ultra thick, opacity=0.3] (-0.6,3) ellipse (0.35cm and 0.7cm);%1
				\draw[rotate around={90:(-2.4,1.2)}, color=green, fill=green, ultra thick, opacity=0.3] (-2.4,1.2) ellipse (0.35cm and 0.35cm); %3
				\draw[rotate around={90:(2.4,1.2)}, color=cyan, fill=cyan, ultra thick, opacity=0.3] (2.4,1.2) ellipse (0.7cm and 0.35cm); %6
				
				\node [label=left:{$X_2$}] at (-1,3) [circle,draw, fill=black, opacity=1, color=black, inner sep=0.8mm] (i) {};
				\node at (-0.2,3) [circle,draw, fill=black, opacity=1, color=black, inner sep=0.8mm] (h) {};
				\node [label=left:{$X_1$}] at (-2.4,1.2) [circle,draw, fill=black, opacity=1, color=black, inner sep=0.8mm] (e) {};
				\node at (2.4,0.8) [circle,draw, fill=black, opacity=1, color=black, inner sep=0.8mm] (f) {};
				\node [label=above:{$X_3$}] at (2.4,1.6) [circle,draw, fill=black, opacity=1, color=black, inner sep=0.8mm] (g) {};
				
				\draw (a) -- (b) node[midway, above]{};
				\draw (a) -- (c) node[midway, above]{};
				\draw (a) -- (d) node[midway, above]{};
				\draw (b) -- (c) node[midway, above]{};
				\draw (b) -- (d) node[midway, above]{};
				\draw (c) -- (d) node[midway, above]{};
				\draw (c) -- (e) node[midway, above]{};
				\draw (a) -- (e) node[midway, above]{};
				\draw (b) -- (e) node[midway, above]{};
				\draw (f) -- (g) node[midway, above]{};
				\draw (a) -- (f) node[midway, above]{};
				\draw (b) -- (f) node[midway, above]{};
				\draw (a) -- (g) node[midway, above]{};
				\draw (b) -- (g) node[midway, above]{};
				\draw (h) -- (i) node[midway, above]{};
				\draw (h) -- (b) node[midway, above]{};	
				\draw (h) -- (d) node[midway, above]{};

				\draw[rotate around={90:(4,0.4)}, color=yellow, fill=yellow, ultra thick, opacity=0.3] (4,0.4) ellipse (0.25cm and 0.25cm); 
				\node (ee) at (4,0.4) {7};
				\draw[rotate around={90:(4,2)}, color=yellow, fill=yellow, ultra thick, opacity=0.3] (4,2) ellipse (0.25cm and 0.25cm); 
				\node (ff) at (4,2) {5};
				\draw[rotate around={90:(5,1.2)}, color=yellow, fill=yellow, ultra thick, opacity=0.3] (5,1.2) ellipse (0.25cm and 0.25cm); 
				\node (gg) at (5,1.2) {6};
				
				\draw[rotate around={90:(4.8,3)}, color=gray, fill=gray, ultra thick, opacity=0.3] (4.8,3) ellipse (0.35cm and 0.7cm);%4	
				\draw[rotate around={90:(6.2,1.2)}, color=blue, fill=blue, ultra thick, opacity=0.3] (6.2,1.2) ellipse (0.35cm and 0.35cm);%5
				\draw[rotate around={90:(5.3,0)}, color=red, fill=red, ultra thick, opacity=0.3] (5.3,0) ellipse (0.35cm and 0.35cm);%6
				\draw[rotate around={90:(4,-0.7)}, color=teal, fill=teal, ultra thick, opacity=0.3] (4,-0.7) ellipse (0.35cm and 0.35cm);%7
				
				\node [label=right:{$X_4$}] at (5.2,3) [circle,draw, fill=black, opacity=1, color=black, inner sep=0.8mm] (ii) {};
				\node at (4.4,3) [circle,draw, fill=black, opacity=1, color=black, inner sep=0.8mm] (jj) {};
				\node [label=right:{$X_5$}] at (6.2,1.2) [circle,draw, fill=black, opacity=1, color=black, inner sep=0.8mm] (kk) {};
				\node [label=right:{$X_6$}] at (5.3,0) [circle,draw, fill=black, opacity=1, color=black, inner sep=0.8mm] (ll) {};
				\node [label=right:{$X_7$}] at (4,-0.7) [circle,draw, fill=black, opacity=1, color=black, inner sep=0.8mm] (mm) {};
				
				\draw (ee) -- (ff) node[midway, above]{};	
				\draw (gg) -- (ff) node[midway, above]{};	
				\draw (ee) -- (gg) node[midway, above]{};	
				\draw (g) -- (ff) node[midway, above]{};	
				\draw (g) -- (ee) node[midway, above]{};
				\draw (f) -- (ee) node[midway, above]{};	
				\draw (ii) -- (jj) node[midway, above]{};		
				\draw (ii) -- (ff) node[midway, above]{};		
				\draw (ii) -- (gg) node[midway, above]{};
				\draw (jj) -- (ff) node[midway, above]{};		
				\draw (jj) -- (gg) node[midway, above]{};		
				\draw (kk) -- (gg) node[midway, above]{};	
				\draw (ee) -- (ll) node[midway, above]{};		
				\draw (gg) -- (ll) node[midway, above]{};
				\draw (ee) -- (mm) node[midway, above]{};			
				
				\node (11) at (2,-1.1) {(a)};	
				\label{b}
			\end{tikzpicture}
			
			\label{fig:UDVGdT}
		\end{subfigure}
		~
		\begin{subfigure}[t]{0.75\linewidth}
			\centering
			\begin{tikzpicture}[scale=0.75]
				
				\draw[black, dashed, ultra thick] (-3.25,-2.5) -- (-0.75,-4.5) -- (-3.25,-6.5);
				\draw[black, dashed, ultra thick] (-0.75,-4.5) -- (3.25,-4.5);
				\draw[black, dashed, ultra thick] (5.75,-2.5) -- (3.25,-4.5) -- (5.75,-6.5);
				
				\draw[orange, ultra thick] (-2.4,-2.85) -- (-0.25,-4.65) -- (-2.3,-6.25);
				\draw[orange, ultra thick] (-0.25,-4.65) -- (0.9,-4.65);
				\node [label=left:{1}] (1) at (-2,-6.4) { };
				
				\draw[orange, ultra thick] (-2.65,-6.25) -- (-0.9,-4.9) -- (0.9,-4.9);
				\node [label=left:{3}] (3) at (-2.45,-6.4) { };
				
				\draw[orange, ultra thick] (-1.95,-6.25) -- (-0.4,-5.05);
				\node [label=left:{2}] (2) at (-1.55,-6.4) { };
				
				\draw[orange, ultra thick] (-2.4,-2.55) -- (-0.75,-3.95);
				\node [label=left:{4}] (4) at (-0.3,-4) { };
				
				\draw[yellow, ultra thick] (5.15,-6.25) -- (3.4,-4.9) -- (1.3,-4.9);
				\node [label=right:{7}] (7) at (4.95,-6.4) { };
				
				\draw[yellow, ultra thick] (4.5,-3.2) -- (3.4,-4.1) -- (1.75,-4.1);
				\node [label=right:{5}] (5) at (4.3,-3.05) { };
				
				\draw[yellow, ultra thick] (5.5,-3) -- (3.6,-4.5) -- (4.7,-5.35);
				\node [label=right:{6}] (6) at (5.3,-2.85) { };
				
				\node [label=left:{\textcolor{green}{$X_1$}}] (11) at (-2.4,-5.5) { };
				\draw[green, line width=2.6] (-3.4,-6.4) -- (-2.5,-5.7); %X1
				
				\node [label=left:{\textcolor{magenta}{$X_2$}}] (22) at (-2.4,-3.55) { };
				\draw[magenta, line width=2.6] (-3.4,-2.65) -- (-2.3,-3.55); %X2
				
				\node [label=right:{\textcolor{cyan}{$X_3$}}] (33) at (0.5,-4.05) { };
				\draw[cyan, line width=2.6] (0.5,-4.35) -- (2.25,-4.35); %X3
				
				\node [label=left:{\textcolor{gray}{$X_4$}}] (44) at (4.3,-2.8) { };
				\draw[gray, line width=2.6] (4.3,-2.9) -- (3.2,-3.8); %X4
				
				\node [label=right:{\textcolor{blue}{$X_5$}}] (55) at (5.65,-3.25) { };
				\draw[blue, line width=2.6] (6,-2.85) -- (5.2,-3.5); %X5
				
				\node [label=right:{\textcolor{red}{$X_6$}}] (66) at (4.9,-5.15) { };~
				\draw[red, line width=2.6] (4.3,-4.75) -- (5.1,-5.4); %X6
				
				\node [label=right:{\textcolor{teal}{$X_7$}}] (77) at (5.7,-6.15) { };
				\draw[teal, line width=2.6] (5.1,-5.75) -- (5.9,-6.4); %X7
				
				\node (11) at (1.78,-7) {(b)};
				\label{c}
				
			\end{tikzpicture}
			\label{fig:UDVGdTT}
		\end{subfigure}
		
		\caption{(a) A $T$-graph $G$ with its two disjoint maximal cliques $C_1 = \{ 1, 2, 3, 4\}$ and $C_2 = \{ 5, 6, 7\}$. (b) A $T$-representation of $G$ with $C_1$ and $C_2$ placed on the branching nodes.}		
		\label{fig:TGraph}
	\end{figure}
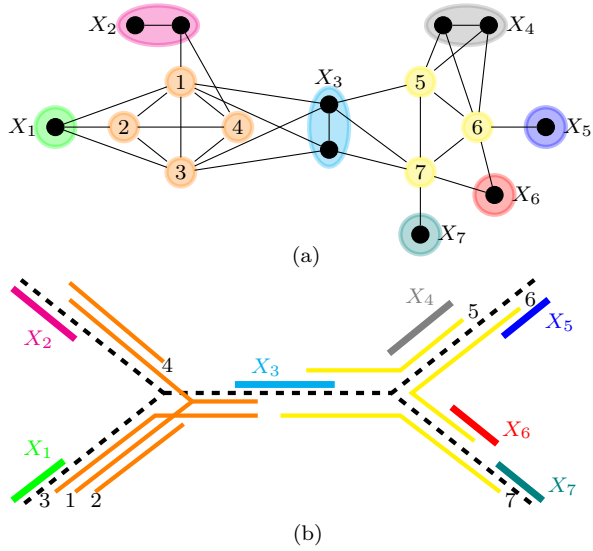
	
	In Figure~\ref{fig:TGraph} (a), we see a $T$-graph $G$ whose disjoint maximal cliques $C_1$ and $C_2$ placed on the branching
	nodes of $T$ are colored orange and yellow, respectively, and the connected components $X_1, \dots, X_7$ of $G-(C_1 \cup C_2)$ are colored differently. All $X_1$, $X_2$ and $X_3$ have non-empty attachments in $C_1$, thus they are placed on incident edges to $b_1$.
	Analogously, all $X_3$, $X_4$, $X_5$, $X_6$ and $X_7$ have non-empty attachments in $C_2$, thus they are placed on incident edges to $b_2$ (for $X_3$, its is the edge~$b_1b_2$). 
	For two $S_d$-instances with the centers $C_1$ and $C_2$, let $P_1$ and $P_2$ denote the partial orders on the connected components. Since $X_3$ has attachments in both $C_1$ and $C_2$,  i.e. $({N_{C_1 \cup C_2}}(X_3) \cap C_1)\setminus C_2 \neq \emptyset$ and $({N_{C_1 \cup C_2}}(X_3) \cap C_2)\setminus C_1 \neq \emptyset$, it is the common component for these $S_d$-instances and appears in both $P_1$ and $P_2$. One of the chain covers of $P_1$ is $(X_1)$, $(X_2)$, $(X_3)$, and one of the chain covers of $P_2$ is $(X_3)$, $(X_4,X_5)$, $(X_6,X_7)$. In Figure~\ref{fig:TGraph} (b), the corresponding $T$-representation of $G$ is given where the connected components are placed on the edges of $T$ according to these chain covers.
	
	Let the {\em leafage} of a chordal graph $G$ be defined as the smallest
	number of leafs of a tree $T$ such that $G$ is a $T$-graph.
	We remark that the leafage can be equivalently defined only from
	the clique graph of~$G$~\cite{mckee1999topics}.
	
	Due to being a superclass of $S_d$-graphs, the isomorphism problem for
	$T$-graphs and proper $T$-graphs is also GI-complete by Proposition~\ref{prop:SdGIc}.

	\section{$S_d$-graph isomorphism parameterized by the clique size}\label{easysection}
	
	We start with an easy case and show that $S_d$-graph isomorphism can be solved in \textbf{FPT}-time when the maximal clique size of given $S_d$-graphs are bounded by a parameter $p$. We emphasize that the complexity of this case does not depend on the width $d$ and we do not need the $S_d$-representations.
	
	Recall that $S_d$-graphs are chordal. Therefore, they have linearly many
	maximal cliques which can be listed in linear time~\cite{recogChordaLinear}. 
	For an $S_d$-graph $G$ with a central clique $C$, the induced subgraph $G[C \cup X_i]$ 
	is an interval graph for each connected component $X_i$ of $G-C$, and the
	isomorphism of interval graphs can be tested in linear time
	\cite{recogIntervalLinear}.  Given two $S_d$-graphs $G$ and $H$,
	we hence compute their collections of maximal cliques $\mathcal{S}$ and $\mathcal{T}$, respectively.
	For each pair of equal-size cliques $C \in\mathcal{S}$ and $D \in\mathcal{T}$ such that all
	connected components in $G-C$ and $H-D$ are interval graphs, we can
	efficiently compare each pair of connected components of $G-C$ and $H-D$ to
	interval graph isomorphism.  If these components can be perfectly
	matched with respect to the isomorphism, we only need to guarantee
	a consistent bijection map between the cliques $C$ and $D$ to correctly conclude
	that $G$ and $H$ are isomorphic.  When the maximal clique size of given $G$ and $H$
	is bounded by a parameter $p$, we can simply loop through all $p!$ possible
	bijections between the cliques $C$ and $D$.
	The fine details are shown in Algorithm~\ref{pseudoEasySD}.
	
	\begin{algorithm} [tbh]
		\caption{Isomorphism test for $S_d$-graphs with bounded clique size ($G$,$H$)}
		\label{pseudoEasySD}\normalsize
		\begin{algorithmic}[1]
			\Require Given two $S_d$-graphs $G$ and $H$ (their representations not required),
			the parameters $d$ of $S_d$ 
			and $p$ bounding the maximal clique size of $G$ and~$H$.
			\Ensure Result of the isomorphism test between $G$ and $H$.
			\medskip
			
			\State Find the maximal clique collections $\mathcal{S}$ and 
			$\mathcal{T}$ of the chordal graphs~$G$~and~$H$;
			\If{$|\mathcal{S}|$ $\neq$ $|\mathcal{T}|$, or the
				cardinalities of members of $\mathcal{S}$ and $\mathcal{T}$ do not match}
			\State\Return ``$G$ and $H$ are not isomorphic'';
			\EndIf{}
			
			\Repeat{~for each maximal clique $C$ $\in$ $\mathcal{S}$}:
			\State Find the connected components $X_1,X_2,\ldots,X_k$ of $G-C$,
			assuming the \mbox{equivalent} connected components
			are joined into single bridge(s) of~$C$;
			\Until{all $G[C \cup X_i]$, $1\leq i\leq k$, are interval graphs};
			\label{li:fixC}
			
			\State Fix any bijective labeling $C\to\{1,2,\ldots,|C|\}$ on the vertices of $C$;
			\For{each $D$ $\in$ $\mathcal{T}$ with $\vert C \vert = \vert D \vert$} 
			
			\State Find the connected components $Y_1,Y_2,\ldots,Y_l$ of $H-D$,
			assuming the \mbox{equivalent} connected components
			are joined into single bridge(s) of~$D$;
			
			\If{$k = l$, and all $H[D \cup Y_j]$, $1\leq j\leq k$, are interval graphs}
			
			\For{each bijective  $D\to\{1,2,\ldots,|D|\}$ on the vertices of $D$}	
			\For{each pair $X_i$ and $Y_j$, $1\leq i,j\leq k$}
			\State Compare the interval graphs $G[C \cup X_i]$
			and $H[D \cup Y_j]$ to isomorphism, respecting the labels of $C$ and $D$;
			\label{li:isopair}
			\State Mark isomorphic pairs of them as  ``symmetric'';
			\EndFor{}
			
			\State $\mathcal{Y} \leftarrow \{Y_1,Y_2,\ldots,Y_k\}$;
			\For{each $X_i$, $1\leq i\leq k$}
			\State Find greedily $Y_j \in \mathcal{Y}$ symmetric to
			$X_i$, and delete~$Y_j$~from~$\mathcal{Y}$;
			\EndFor{}
			
			\If{successful for all $X_i$, $1\leq i\leq k$}
			\State\Return ``$G$ and $H$ are isomorphic'';
			\EndIf{}				
			
			\EndFor{}	
			\EndIf{}
			\EndFor{}
			\State\Return ``$G$ and $H$ are not isomorphic'';
		\end{algorithmic}
	\end{algorithm}

	\paragraph*{The complexity analysis of Algorithm~\ref{pseudoEasySD}.}
	~
	Let $n$ be the number of vertices of both $G$ and $H$.
	The maximal clique collections $\mathcal{S}$ and $\mathcal{T}$ of them
	(each of length~$\leq n$) are found using simplicial vertex elimination
	in time $\mathcal{O}(pn)$.
	Comparing the cardinalities of the members of $\mathcal{S}$ and $\mathcal{T}$
	takes $\mathcal{O}(n\log n)$ time.  
	Finding the connected components of $G-C$ (also of $H-D$) takes linear time
	and testing for interval graphs is also in linear time~\cite{recogIntervalLinear},
	and so a suitable clique $C\in\mathcal{S}$ can be fixed in time $\mathcal{O}(n^3)$.
	
	Then we loop through $D\in\mathcal{T}$ and all labelings of $D$, in the
	worst case, which means at most $p!\cdot n$ iterations.
	Comparing the interval graphs $G[C \cup X_i]$ and $H[D \cup Y_j]$ to 
	isomorphism takes linear time for each pair $X_i,Y_j$, but we can
	do a better runtime analysis since the following holds:
	\begin{equation*}
		\sum_{i=1}^k|C\cup X_i|\leq |C|n+n\leq pn+n
	\end{equation*}
	Therefore, comparing all the pairs of graphs
	$G[C \cup X_i]$ and $H[D \cup Y_j]$ takes altogether $\mathcal{O}(pn^2)$.
	Checking for a perfect matching of symmetric (i.e., isomorphic) pairs
	among them is then trivially in time $\mathcal{O}(n^2)$.
	
	Thus, the overall complexity of this algorithm is 
	\begin{equation*}
		\mathcal{O}(n^3 + p!\cdot n\cdot(pn^2 + n^2)) = \mathcal{O}(p!\cdot pn^3)
	\end{equation*}
	which belongs to \textbf{FPT} with respect to $p$.

	\begin{restatable}{theorem}{prooff}
		\label{theo:SDSMallTheo}
		
		Algorithm~\ref{pseudoEasySD} correctly decides whether two $S_d$-graphs are isomorphic in \emph{FPT}-time parameterized by the maximal
		clique size $p$.
	\end{restatable}
	
	\begin{proof}
		
		Assume that Algorithm~\ref{pseudoEasySD} returns ``$G$ and $H$ are isomorphic''.
		Since the fixed labelings of $C$ and $D$ are respected when comparing each pair $G[C \cup X_i]$ and $H[D \cup Y_j]$ to interval graph isomorphism, we indeed have that~$G\simeq H$.
		Note that this conclusion holds true regardless of whether the fixed clique $C$ is
		the central clique in some $S_d$-representation of $G$.
		
		Let, on the other hand, $G$ and $H$ be two isomorphic $S_d$-graphs 
		with cliques of size~$\leq p$, and $f$ be their isomorphism.
		Since $G$ has an $S_d$-graph representation, an admissible clique $C$ must
		be found on line \ref{li:fixC} of the algorithm (as the central clique
		of the representation, if not before).
		Then, for the fixed labeling of $C$ and the image of this labeling in $D$
		under~$f$, the isomorphism tests between $G[C \cup X_i]$ and $H[D \cup Y_j]$ where $Y_j=f(X_i)$
		on line \ref{li:isopair} must succeed.
		Moreover, since isomorphism is an equivalence relation, the
		subsequent greedy pairing of symmetric components succeeds as well.
		Hence the algorithm returns ``$G$ and $H$ are isomorphic'', as desired.\qed
	\end{proof}

	\section{Reduction to $S_d$-graph isomorphism from posets of width $d$}\label{reductionsection}
	
	In Section~\ref{easysection}, we focused on an easy case of
	$S_d$-graph isomorphism which can be tested in \textbf{FPT}-time parameterized by
	the maximal clique size $p$;
	by greedily trying all $p!$ labelings of the (suitably chosen) central cliques.
	In general, the maximal clique size can grow up to $\Omega(n)$,
	and thus we now cannot afford to try all their possible labelings.
	As suggested by Proposition~\ref{prop:SdGIc}, we must take advantage of the
	parameter~$d$, that is, of the underlying structure of $S_d$-representations of $G$ and~$H$
	(as we will see, we do not need a particular representation given for that).
	
	According to the characterization of Chaplick et al.\ from Proposition~\ref{prop:Sdcharacteriz},
	it looks useful to study the isomorphism of the central posets of $G$ and
	$H$ (on their chosen central cliques) in order to decide isomorphism of $G$ and $H$ themselves.
	Before we do this in Section~\ref{hardsection}, we show that it is indeed
	necessary to be able to solve isomorphism of posets of width $d$;
	by giving the following polynomial-time reduction and subsequent Theorem~\ref{theo:red2}.
	
	Given a poset $P$ of width $d$ on $n$ elements,
	we construct an $S_d$-graph $G$ from $P$ (in polynomial time) as follows:
	\begin{enumerate}
		\item Model $P$ by the set inclusion between the sets $M_1, \dots, M_n$,
		where each $M_i$ consists of all comparable elements with $i$ from the lower
		levels and itself.  Formally, $M_i = \{j \in P: j\preceq_P i\}$ for all~$i\in P$.
		\item Take the union $M=\bigcup_{i \in P} M_i$ and form the central
		clique $C$ of size $|M|+2$ by adding $|M|$ vertices corresponding to
		the elements of $M$, two dummy vertices, and all edges~on~$C$. 
		Note that the two dummy vertices are added to $C$ to mark $C$ as the
		unique maximum-size clique $C$ of~$G$.
		
		\item For each $M_i$, add a new vertex $v_i$ to $G$ adjacent
		exactly to the subset of vertices in $C$ corresponding to $M_i$.
	\end{enumerate}
	
	Figure~\ref{fig:reduction1} is an illustration of this reduction. In
	Figure~\ref{fig:reduction1} (a), we have a poset $P$ of width $d=3$ whose
	elements are $1$, $2$, $3$, $4$, $5$, $6$, $7$, $8$ and $9$.  Using the
	above construction, we model the poset $P$ by the set inclusion with sets
	$M_1,\ldots,M_9$ as in Figure~\ref{fig:reduction1} (b),
	and then we construct the graph $G$ in Figure~\ref{fig:reduction1}
	(c), where the cyan set consists of all vertices of the unique maximum
	clique $C$ (assuming that there is an edge between each pair of vertices in this set).
	Each connected component of $G-C$ is a single vertex denoted
	by $v_1, \dots, v_9$ and adjacent to exactly the vertices from $M_1,\ldots,M_9$, respectively.
	
	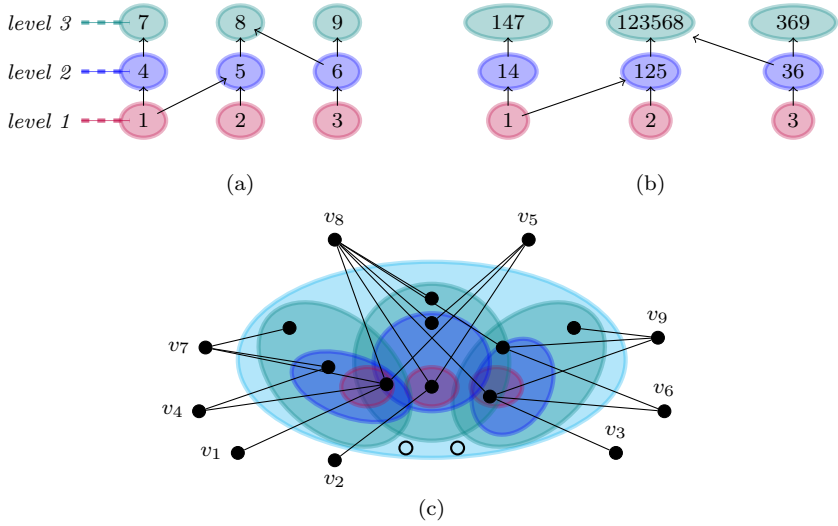
\begin{figure}[tbp]
		\centering
		\begin{tikzpicture}[xscale=0.85, yscale=0.65]

			\draw[rotate around={90:(-10.2,-1.2)}, color=purple, fill=purple, ultra thick, opacity=0.3] (-10.2,-1.2) ellipse (0.35cm and 0.35cm); % 1 
			\draw[rotate around={90:(-8.7,-1.2)}, color=purple, fill=purple, ultra thick, opacity=0.3] (-8.7,-1.2) ellipse (0.35cm and 0.35cm); % 2 
			\draw[rotate around={90:(-7.2,-1.2)}, color=purple, fill=purple, ultra thick, opacity=0.3] (-7.2,-1.2) ellipse (0.35cm and 0.35cm); % 3 
			
			\draw[rotate around={90:(-10.2,-0.2)}, color=blue, fill=blue, ultra thick, opacity=0.3] (-10.2,-0.2) ellipse (0.35cm and 0.35cm);% 4
			\draw[rotate around={90:(-8.7,-0.2)}, color=blue, fill=blue, ultra thick, opacity=0.3] (-8.7,-0.2) ellipse (0.35cm and 0.35cm);% 5
			\draw[rotate around={90:(-7.2,-0.2)}, color=blue, fill=blue, ultra thick, opacity=0.3] (-7.2,-0.2) ellipse (0.35cm and 0.35cm);% 6 
			
			\draw[rotate around={50:(-10.2,0.8)}, color=teal, fill=teal, ultra thick, opacity=0.3] (-10.2,0.8) ellipse (0.35cm and 0.35cm); % 7
			\draw[rotate around={50:(-8.7,0.8)}, color=teal, fill=teal, ultra thick, opacity=0.3] (-8.7,0.8) ellipse (0.35cm and 0.35cm); % 8 
			\draw[rotate around={-50:(-7.2,0.8)}, color=teal, fill=teal, ultra thick, opacity=0.3] (-7.2,0.8) ellipse (0.35cm and 0.35cm); % 9 
			
			\node (1) at (-10.2,-1.2) {$1$};
			\node (2) at (-8.7,-1.2) {$2$};
			\node (3) at (-7.2,-1.2) {$3$};
			\node (4) at (-10.2,-0.2) {$4$};
			\node (5) at (-8.7,-0.2) {$5$};
			\node (6) at (-7.2,-0.2) {$6$};
			\node (7) at (-10.2,0.8) {$7$};
			\node (8) at (-8.7,0.8) {$8$};
			\node (9) at (-7.2,0.8) {$9$};
			
			\node (100) at (-11.8,-1.2) {\emph{level 1}};
			\draw[purple, fill=purple, opacity=0.6, ultra thick, dashed] (100)--(1);
			
			\node (200) at (-11.8,-0.2) {\emph{level 2}};
			\draw[blue, fill=blue, opacity=0.6, ultra thick, dashed] (200)--(4);
			
			\node (300) at (-11.8,0.8) {\emph{level 3}};
			\draw[teal, fill=teal, opacity=0.6, ultra thick, dashed] (300)--(7);
			
			\draw[->] (1)--(4) node[midway, above]{};
			\draw[->] (6)--(8) node[midway, above]{};
			\draw[->] (1)--(5) node[midway, above]{};
			\draw[->] (2)--(5) node[midway, above]{};
			\draw[->] (3)--(6) node[midway, above]{};
			\draw[->] (4)--(7) node[midway, above]{};			
			\draw[->] (5)--(8) node[midway, above]{};	
			\draw[->] (6)--(9) node[midway, above]{};
			
			\node (11) at (-8.7,-2.5) {(a)};
			\label{aaa}
		\end{tikzpicture}
		\qquad\qquad
		\begin{tikzpicture}[xscale=0.75, yscale=0.65]

			\draw[rotate around={90:(-11.2,-1.2)}, color=purple, fill=purple, ultra thick, opacity=0.3] (-11.2,-1.2) ellipse (0.35cm and 0.35cm); % 1 
			\draw[rotate around={90:(-8.7,-1.2)}, color=purple, fill=purple, ultra thick, opacity=0.3] (-8.7,-1.2) ellipse (0.35cm and 0.35cm); % 2 
			\draw[rotate around={90:(-6.2,-1.2)}, color=purple, fill=purple, ultra thick, opacity=0.3] (-6.2,-1.2) ellipse (0.35cm and 0.35cm); % 3  	
			\draw[rotate around={90:(-11.2,-0.2)}, color=blue, fill=blue, ultra thick, opacity=0.3] (-11.2,-0.2) ellipse (0.35cm and 0.5cm);% 4
			\draw[rotate around={90:(-8.7,-0.2)}, color=blue, fill=blue, ultra thick, opacity=0.3] (-8.7,-0.2) ellipse (0.35cm and 0.5cm);% 5
			\draw[rotate around={90:(-6.2,-0.2)}, color=blue, fill=blue, ultra thick, opacity=0.3] (-6.2,-0.2) ellipse (0.35cm and 0.5cm);% 6  	
			\draw[rotate around={90:(-11.2,0.8)}, color=teal, fill=teal, ultra thick, opacity=0.3] (-11.2,0.8) ellipse (0.35cm and 0.75cm); % 7
			\draw[rotate around={90:(-8.7,0.8)}, color=teal, fill=teal, ultra thick, opacity=0.3] (-8.7,0.8) ellipse (0.35cm and 0.75cm); % 8 
			\draw[rotate around={90:(-6.2,0.8)}, color=teal, fill=teal, ultra thick, opacity=0.3] (-6.2,0.8) ellipse (0.35cm and 0.75cm); % 9  	
			\node (1) at (-11.2,-1.2) {$1$};
			\node (2) at (-8.7,-1.2) {$2$};
			\node (3) at (-6.2,-1.2) {$3$};
			\node (4) at (-11.2,-0.2) {$14$};
			\node (5) at (-8.7,-0.2) {$125$};
			\node (6) at (-6.2,-0.2) {$36$};
			\node (7) at (-11.2,0.8) {$147$};
			\node (8) at (-8.7,0.8) {$123568$};
			\node (9) at (-6.2,0.8) {$369$};
			
			\draw[->] (1)--(4) node[midway, above]{};
			\draw[->] (6)--(8) node[midway, above]{};
			\draw[->] (1)--(5) node[midway, above]{};
			\draw[->] (2)--(5) node[midway, above]{};
			\draw[->] (3)--(6) node[midway, above]{};
			\draw[->] (4)--(7) node[midway, above]{};			
			\draw[->] (5)--(8) node[midway, above]{};	
			\draw[->] (6)--(9) node[midway, above]{};
			
			\node (11) at (-8.7,-2.5) {(b)};
		\end{tikzpicture}

		\begin{tikzpicture}[xscale=0.85, yscale=0.65]
			\draw[rotate around={90:(-2,0.55)}, color=cyan, fill=cyan, ultra thick, opacity=0.3] (-2,0.55) ellipse (2cm and 3cm); % 7
			
			\draw[rotate around={150:(-0.5,0.25)}, color=teal, fill=teal, ultra thick, opacity=0.3] (-0.5,0.25) ellipse (1cm and 1.6cm); % 7
			\draw[rotate around={165:(-0.75,0)}, color=blue, fill=blue, ultra thick, opacity=0.3] (-0.75,0) ellipse (0.6cm and 1cm); % 4
			\draw[rotate around={90:(-1,0)}, color=purple, fill=purple, ultra thick, opacity=0.3] (-1,0) ellipse (0.4cm and 0.4cm); % 1	
			\draw[rotate around={90:(-2,0.5)}, color=teal, fill=teal, ultra thick, opacity=0.3] (-2,0.5) ellipse (1.6cm and 1.2cm); % 8
			\draw[rotate around={90:(-2,0.5)}, color=blue, fill=blue, ultra thick, opacity=0.3] (-2,0.5) ellipse (1cm and 0.9cm); % 5
			\draw[rotate around={90:(-2,0)}, color=purple, fill=purple, ultra thick, opacity=0.3] (-2,0) ellipse (0.4cm and 0.4cm); % 2	
			\draw[rotate around={30:(-3.5,0.25)}, color=teal, fill=teal, ultra thick, opacity=0.3] (-3.5,0.25) ellipse (1cm and 1.6cm); % 9
			\draw[rotate around={60:(-3.25,0)}, color=blue, fill=blue, ultra thick, opacity=0.3] (-3.25,0) ellipse (0.6cm and 1cm); % 6
			\draw[rotate around={90:(-3,0)}, color=purple, fill=purple, ultra thick, opacity=0.3] (-3,0) ellipse (0.4cm and 0.4cm); % 3
			
			\node at (-2.7,0.05) [circle,draw, fill=black, opacity=1, color=black, inner sep=0.6mm] (1) {};
			\node[label=left:{$v_1$}] at (-5,-1.35) [circle,draw, fill=black, opacity=1, color=black, inner sep=0.6mm] (11) {};
			
			\node at (-3.6,0.4) [circle,draw, fill=black, opacity=1, color=black, inner sep=0.6mm] (4) {};
			\node[label=left:{$v_4$}] at (-5.6,-0.5) [circle,draw, fill=black, opacity=1, color=black, inner sep=0.6mm] (44) {};
			
			\node at (-4.2,1.2) [circle,draw, fill=black, opacity=1, color=black, inner sep=0.6mm] (7) {};
			\node[label=left:{$v_7$}] at (-5.5,0.8) [circle,draw, fill=black, opacity=1, color=black, inner sep=0.6mm] (77) {};
			
			\node at (-2,0) [circle,draw, fill=black, opacity=1, color=black, inner sep=0.6mm] (2) {};
			\node[label=below:{$v_2$}] at (-3.5,-1.5) [circle,draw, fill=black, opacity=1, color=black, inner sep=0.6mm] (22) {};
			
			\node at (-2,1.3) [circle,draw, fill=black, opacity=1, color=black, inner sep=0.6mm] (5) {};
			\node[label=above:{$v_5$}] at (-0.5,3) [circle,draw, fill=black, opacity=1, color=black, inner sep=0.6mm] (55) {};
			
			\node at (-2,1.8) [circle,draw, fill=black, opacity=1, color=black, inner sep=0.6mm] (8) {};
			\node[label=above:{$v_8$}] at (-3.5,3) [circle,draw, fill=black, opacity=1, color=black, inner sep=0.6mm] (88) {};
			
			\node at (-1.1,-0.2) [circle,draw, fill=black, opacity=1, color=black, inner sep=0.6mm] (3) {};
			\node[label=above:{$v_3$}] at (0.85,-1.35) [circle,draw, fill=black, opacity=1, color=black, inner sep=0.6mm] (33) {};
			
			\node at (-0.9,0.8) [circle,draw, fill=black, opacity=1, color=black, inner sep=0.6mm] (6) {};
			\node[label=above:{$v_6$}] at (1.6,-0.5) [circle,draw, fill=black, opacity=1, color=black, inner sep=0.6mm] (66) {};
			
			\node at (0.2,1.2) [circle,draw, fill=black, opacity=1, color=black, inner sep=0.6mm] (9) {};
			\node[label=above:{$v_9$}] at (1.5,1) [circle,draw, fill=black, opacity=1, color=black, inner sep=0.6mm] (99) {};
			
			\node at (-2.4,-1.25) [circle,draw, fill=none, thick, color=black, inner sep=0.6mm] {};
			\node at (-1.6,-1.25) [circle,draw, fill=none, thick, color=black, inner sep=0.6mm] {};
			
			\draw[-] (1) -- (11) node[midway, above]{};
			\draw[-] (44) -- (1) node[midway, above]{};
			\draw[-] (44) -- (4) node[midway, above]{};
			\draw[-] (1) -- (77) node[midway, above]{};
			\draw[-] (4) -- (77) node[midway, above]{};
			\draw[-] (7) -- (77) node[midway, above]{};
			\draw[-] (2) -- (22) node[midway, above]{};
			\draw[-] (55) -- (2) node[midway, above]{};
			\draw[-] (55) -- (1) node[midway, above]{};
			\draw[-] (55) -- (5) node[midway, above]{};
			\draw[-] (88) -- (1) node[midway, above]{};	
			\draw[-] (88) -- (6) node[midway, above]{};
			\draw[-] (88) -- (3) node[midway, above]{};	
			\draw[-] (88) -- (2) node[midway, above]{};
			\draw[-] (88) -- (5) node[midway, above]{};
			\draw[-] (88) -- (8) node[midway, above]{};
			\draw[-] (33) -- (3) node[midway, above]{};
			\draw[-] (66) -- (3) node[midway, above]{};
			\draw[-] (66) -- (6) node[midway, above]{};
			\draw[-] (99) -- (3) node[midway, above]{};
			\draw[-] (99) -- (6) node[midway, above]{};
			\draw[-] (99) -- (9) node[midway, above]{};
			
			\node (111) at (-2,-2.5) {(c)};
			\label{bbb}	
		\end{tikzpicture}
		\caption{(a) An example poset $P$ (as a Hasse diagram), and its levels. 
			(b) A representation of this poset by the set inclusion.
			(c) The corresponding $S_d$-graph $G$ where the cyan set corresponds to
			the central maximal clique.}	
		\label{fig:reduction1}	
	\end{figure}
	
	\begin{lem}\label{lem:posettoSd}
		The graph $G$ constructed from a poset $P$ of width $d$ using the above reduction is an $S_d$-graph.
	\end{lem}
	
	\begin{proof}
		$C$ is the maximal (and also maximum) clique of $G$, represented in the
		branching node of a subdivision of $S_d$.
		A chain cover of $P$ of size $\leq d$ determines
		an ordered distribution of the vertices $V(G-C)=\{v_1,\ldots,v_n\}$ to the $d$ rays of the representation, 
		such that their adjacent vertices in $C$ form a chain by inclusion on each ray.
		Hence this arrangement is realizable as interval graphs on the $d$ rays of a subdivision of $S_d$.
		\qed\end{proof}

	\begin{restatable}{theorem}{redd}
		\label{theo:red2}
		The isomorphism problem of (colored) posets of width $d$ reduces in
		polynomial time to the isomorphism problem of $S_d$-graphs.
	\end{restatable}
	
	\begin{proof} 
		Let $P$ and $Q$ with $n$ elements be two posets of width $d$, and
		$G$ and $H$ be the $S_d$-graphs (cf.~Lemma~\ref{lem:posettoSd}) formed from $P$ and $Q$,
		respectively, by the construction described above.
		
		Assume that $P$ and $Q$ are isomorphic under a bijection $f:P\to Q$.
		Then, $f$ directly defines a bijection from the central clique $C$
		to $D$, mapping the dummy vertices of $C$ to those of $D$ in any order, and a
		bijection from $V(G-C)=\{v_1,\ldots,v_n\}$ to $V(H-D)= \{w_1,\ldots,w_n\}$.
		The composition of these two mappings from $C$ to $D$ and $G-C$ to $H-D$ is clearly a graph isomorphism between $G$ and~$H$.
		
		Assume that $G$ and $H$ are isomorphic under a bijection $g:V(G)\to V(H)$.
		Since $C \subseteq V(G)$ and $D \subseteq V(H)$ are the unique maximum cliques, we have that $g(C)=D$.
		Therefore, a restriction of $g$ is a bijection from
		$V(G-C)=\{v_1,\ldots,v_n\}$ to $V(H-D)=\{w_1,\ldots,w_n\}$.
		With respect to the construction of $G$ and $H$, this restriction induces a bijection
		between the sets $\{M_1,\ldots,M_n\}$ of $P$ (as represented by
		$v_1,\ldots,v_n$) and the sets $\{N_1,\ldots,N_n\}$ of $Q$ (as represented
		by $w_1,\ldots,w_n$) and, in turn, a bijection $g':P\to Q$.
		If $a\preceq_P b$, then $M_a\subseteq M_b$ and the neighborhood of $v_a$ in
		$G$ is included in the neighborhood of $v_b$.
		Then the neighborhood of $g(v_a)$ in $H$ is included in the neighborhood of
		$g(v_b)$ and, regarding to the induced bijection $g'$, 
		$N_{g'(a)}\subseteq N_{g'(b)}$ and $g'(a)\preceq_Q g'(b)$.
		The converse implication holds the same way, and so the posets $P$ and $Q$ are isomorphic.
		
		Lastly, we remark that if the input posets $P$ and $Q$ are given with colored elements, 
		and we are looking for color-preserving isomorphism, we can use the same approach.
		Let, say, the given distinct colors be denoted by $c_1,\ldots,c_r$ where $r\in\mathcal{O}(n)$.
		We then use the above reduction with the following changes:
		(i) If $j$ is a poset element of color $c_i$, then we represent $M_j$ in $G$ not by a single vertex $v_j$,
		but by a copy of the clique $K_i$. Therefore, only the mappings between the vertices contained in equal sized cliques are allowed corresponding to the color-preserving isomorphisms.
		(ii) The number of dummy vertices in this case is $r+1$ in order to mark $C$ of size $n+r+1$ as the unique maximum clique of $G$.
		\qed\end{proof}

	\section{Isomorphism of $S_d$-graphs in general, parameterized by~$d$}\label{hardsection}
	
	When $d$ is a part of the input, the isomorphism problem for $S_d$-graphs is
	GI-complete (Proposition~\ref{prop:SdGIc}), and we proved that $S_d$-graph
	isomorphism can be solved in \textbf{FPT}-time parameterized by the maximal
	clique size (Theorem~\ref{theo:SDSMallTheo}).  In this section, we consider
	$S_d$-graphs without bounding the clique size, and give an
	\textbf{FPT}-time algorithm solving their isomorphism parameterized by $d$.
	
	We first recall the notion of the \emph{automorphism group} which is closely
	related to the graph isomorphism problem.  An \emph{automorphism} is an
	isomorphism of a graph $G$ to itself, and the \emph{automorphism group} of
	$G$ is the group $Aut(G)$ of all automorphisms of $G$.  There exists an isomorphism
	from $G$ to $H$ if and only if the automorphism group of the disjoint union
	$G \uplus H$ contains a permutation exchanging the vertex sets of $G$ and $H$.  
	In fact, assuming connectivity of the graphs $G$ and $H$,
	it is enough to look for a permutation mapping some vertex of $G$ to a
	vertex of $H$, and only among generators of the automorphism group.
	For further details regarding the automorphism groups, see e.g., \cite{furst}.
	
	The elements of any poset $R$ can be partitioned into {\em levels}
	$L_i\subseteq R$ where $i\geq1$; $L_1$ is formed by the minimal elements of
	$R$, and $L_{i+1}$ is inductively formed by the minimal elements of
	$R\setminus(L_1\cup\ldots\cup L_i)$.  
	
	Consider now two $S_d$-graphs $G$ and $H$.
	As forwarded in Section~\ref{reductionsection}, we will approach the isomorphism problem for
	$G$ and $H$ via the automorphism group of the union of the underlying colored central posets of bounded width,
	whose elements are colored regarding the isomorphism types of their interval bridges in $G$ and~$H$.
	However, we will also need to ensure that an automorphism on the central
	posets is indeed consistent with some permutation on the union of the central cliques.
	In contrast to the exhaustive approach used in Section~\ref{easysection},
	we will utilize an involved group-computing approach respecting the attachment sets
	of bridges on $C\cup D$ (so-called cardinality Venn diagram on $C\cup D$).
	
	We first give the following overview of our approach, and then we discuss it in more detail
	(cf. Algorithms~\ref{alg:Gammaprime} and~\ref{alg:fullISO}).
	
	\begin{proc}\label{proc:isomore}\rm~
		Given two $S_d$-graphs $G$ and $H$ on $n$ vertices with
		central cliques $C \subseteq G$ and $D \subseteq H$ such that $|C|=|D|$, let $\mathcal{X}=\{X_1,\ldots,X_k\}$ and $\mathcal{Y}=\{Y_1,\ldots,Y_k\}$ be the sets of connected components in $G-C$ and $H-D$, respectively.
		Assume (as in Algorithm~\ref{pseudoEasySD}) that the equivalent connected components are joined into single bridge(s) of $C$ and $D$, respectively.
		Let $K=G\uplus H$ be the disjoint union of our graphs.
		For $Z\in\mathcal{X}\cup\mathcal{Y}$, let $K(Z)$ denote
		the induced subgraph $G[C\cup Z]$ if $Z\in\mathcal{X}$, and the induced subgraph $H[D\cup Z]$ otherwise. 
		Call an isomorphism of $K(Z)$ to $K(Z')$ {\em respectful} if it maps
		$V(K(Z))\cap(C\cup D)$ to $V(K(Z'))\cap(C\cup D)$.
		We test the existence of an isomorphism between $G$ and $H$ mapping $C$ to $D$ as follows:
		\begin{enumerate}	
			\item Construct the central posets $P$ and $Q$ on $\mathcal{X}$ and $\mathcal{Y}$, 
			respectively, as described in Section~\ref{introSD}.
			Make the disjoint union $R=P\uplus Q$, and compute the levels of~$R$.
			For each pair of components $Z,Z'\in\mathcal{X}\cup\mathcal{Y}$
			on the same level of $R$, compare $K(Z)$
			and $K(Z')$ to respectful isomorphism of interval graphs, and color $Z$ and $Z'$ with respect to the isomorphism type.
			That is, $Z$ and $Z'$ receive the same color if and only if $Z$ and $Z'$ are
			on the same level and $K(Z)\simeq K(Z')$ with a respectful isomorphism.
			
			\item Respecting the colors by the previous step, compute the color-preserving automorphism group $\Gamma$ of $R$. We use Theorem~\ref{thm:bdcm} and Corollary~\ref{cor:bdcm} here.
			
			\item Compute the subgroup $\Gamma'$ of $\Gamma$, consisting of
			those automorphisms $\varrho$ of $R$ for which there exists a permutation
			$f_{\varrho}$ of the set $C\cup D$ such that the following holds;
			for every component $Z\in\mathcal{X}\cup\mathcal{Y}$,
			there is a respectful isomorphism from $K(Z)$ to $K(\varrho(Z))$
			whose restriction to the intersection with $C\cup D$ equals 
			the respective restriction of~$f_{\varrho}$.
			We use Lemma~\ref{lem:preciseVenn}, and Theorem~\ref{thm:furstgen} with
			Corollary~\ref{cor:d-tuples} here.
			
			\item If $P$ and $Q$ are swapped in some automorphism from $\Gamma'$,
			then return that $G$ and $H$ are {\em isomorphic}.
		\end{enumerate}
	\end{proc}
	
	If the above procedure does not say that $G$ and $H$ are isomorphic for
	any pair of maximal cliques $C \subseteq G$ and $D \subseteq H$,
	we return that $G$ and $H$ are {\em non-isomorphic}.
	
	In step 1 of Procedure~\ref{proc:isomore}, we construct the central posets $P$ and $Q$, 
	the disjoint union $R=P\uplus Q$, and its levels. 
	Then, we compute colors on $R$ using the interval graph isomorphism algorithm of~\cite{recogIntervalLinear}. 
	All these are easily executed in polynomial time.
	We will show in detail how to achieve an \textbf{FPT}-time implementation of the other steps in the following subsections.

	\subsection{Computing the automorphism group of a poset of width $d$}
	
	Note that poset isomorphism problem is GI-complete \cite{isoPoset}. 
	However, posets of width $d$ are special and can be handled using the
	following classical concept of bounded color multiplicity.
	
	A \emph{$d$-bounded color multiplicity graph} is a graph $G$ whose 
	vertex set is arbitrarily partitioned into $k$ color classes 
	$V(G)=V_1\cup\ldots\cup V_k$ such that $V_i \cap V_j = \emptyset$ for all $1 \leq i < j \leq k$.
	The number $k$ of colors is arbitrary, but for all $1\leq i\leq k$, the
	cardinality $|V_i|$, called the multiplicity of $V_i$, is at most $d$. 
	We apply the following classical result in our approach.
	
	\begin{theorem}[Babai~\cite{babai-bdcm}, with 
		Furst, Hopcroft and Luks~\cite{furst-bdcm}]\label{thm:bdcm}
		The color-preserving automorphism group (i.e., the generators of it) of a $d$-bounded color multiplicity graph
		can be determined in \textbf{FPT}-time parameterized by~$d$.
	\end{theorem}
	
	Consider a poset $R$ of width $\leq d$ 
	and the levels $L_1,\ldots,L_k$ of~$R$, where $|L_i|\leq d$ for $1\leq i\leq k$. By having the levels of $R$ as color classes, we can ignore the edge directions since the colors will directly correspond to the levels (from a vertex of the color corresponding to the lower level to the vertex of the color corresponding to the higher level). Therefore, a poset of width $d$ is a $d$-bounded multiplicity graph.
	Note that any automorphism of $R$ preserves its levels and this corresponds to the color-preserving automorphisms of a bounded multiplicity graph. Hence we have the following corollary:
	
	\begin{cor}[\cite{babai-bdcm,furst-bdcm}]\label{cor:bdcm}
		The automorphism group of a poset $R$ of width $d$ can be determined in \textbf{FPT}-time parameterized by~$d$.
	\end{cor}
	
	By Corollary~\ref{cor:bdcm}, we can finish step 2 of
	Procedure~\ref{proc:isomore} in \textbf{FPT}-time parameterized by $d$. 
	However, this is not all.
	Having isomorphic central posets only ensures a kind of structural
	equivalence between the chosen central cliques of our graphs $G$ and $H$, regardless of
	the existence of a bijection between the elements of the central cliques.
	
	In Figure~\ref{fig:sharedelements}, as an illustration example, we have two isomorphic
	colored posets $P$ and $Q$ obtained from graphs $G$ and $H$ with the central
	clique $C=D=\{1,2,\ldots,8\}$ (shaded gray), such that the components of
	$G-C$ (of $H-C$) are $9$ singleton vertices having attachments in exactly
	the listed vertices of~$C$ (of~$D$).  The colors are fully determined by the
	cardinality of the shown attachments.  In step 2, $P$ and $Q$ can be swapped
	respecting inclusion and the colors.  However, the components of $G-C$ with
	blue attachments in $P$ (in Figure~\ref{fig:sharedelements}, it is the
	second level marked with an arrow from the right)
	have a common neighbor ($3$) while the analogous two
	components of $H-D$ with blue attachments in $Q$ have no common neighbor. 
	Therefore, the cardinalities of the intersections of blue attachments in $P$
	and in $Q$ differ, which means that there is no bijection between $C$ and $D$
	which maps these attachments to each other, and the corresponding graphs $G$
	and $H$ are hence not isomorphic.
	
	\begin{figure}[tb]
		\centering
		\begin{tikzpicture}[scale=0.8]

			\draw[rotate around={90:(-10.9,-1.2)}, color=purple, fill=purple, ultra thick, opacity=0.3] (-10.9,-1.2) ellipse (0.35cm and 0.35cm); % 1 
			\draw[rotate around={90:(-9.5,-1.2)}, color=purple, fill=purple, ultra thick, opacity=0.3] (-9.5,-1.2) ellipse (0.35cm and 0.35cm); % 2 	
			\draw[rotate around={90:(-7.9,-1.2)}, color=purple, fill=purple, ultra thick, opacity=0.3] (-7.9,-1.2) ellipse (0.35cm and 0.35cm); % 3 
			\draw[rotate around={90:(-6.5,-1.2)}, color=purple, fill=purple, ultra thick, opacity=0.3] (-6.5,-1.2) ellipse (0.35cm and 0.35cm); % 4 
			
			\draw[rotate around={0:(-10.2,-0.2)}, color=blue, fill=blue, ultra thick, opacity=0.3] (-10.2,-0.2) ellipse (0.5cm and 0.35cm);% 5
			\draw[rotate around={0:(-7.2,-0.2)}, color=blue, fill=blue, ultra thick, opacity=0.3] (-7.2,-0.2) ellipse (0.5cm and 0.35cm);% 6 
			
			\draw[rotate around={0:(-10.2,0.8)}, color=green, fill=green, ultra thick, opacity=0.3] (-10.2,0.8) ellipse (0.5cm and 0.35cm); % 7
			\draw[rotate around={0:(-7.2,0.8)}, color=green, fill=green, ultra thick, opacity=0.3] (-7.2,0.8) ellipse (0.5cm and 0.35cm); % 9 
			
			\draw[rotate around={0:(-8.7,2)}, color=cyan, fill=cyan, ultra thick, opacity=0.3] (-8.7,2) ellipse (1cm and 0.35cm); % 10 
			
			\node (1) at (-10.9,-1.2) {$1$};
			\node (2) at (-9.5,-1.2) {$2$};
			\node (3) at (-7.9,-1.2) {$5$};
			\node (4) at (-6.5,-1.2) {$6$};
			
			\node (5) at (-10.2,-0.2) {$123$};
			\node (6) at (-7.2,-0.2) {$356$};
			\node (7) at (-10.2,0.8) {$1234$};
			\node (9) at (-7.2,0.8) {$3456$};
			\node (10) at (-8.7,2) {$12345678$};
			
			\node (11) at (-8.7,-2.2) {(a) $P$};
			
			\draw[->] (1)--(5) node[midway, above]{};
			\draw[->] (2)--(5) node[midway, above]{};
			\draw[->] (3)--(6) node[midway, above]{};
			\draw[->] (4)--(6) node[midway, above]{};
			\draw[->] (5)--(7) node[midway, above]{};				
			\draw[->] (6)--(9) node[midway, above]{};
			\draw[->] (7)--(10) node[midway, above]{};
			\draw[->] (9)--(10) node[midway, above]{};	
			
			\label{aaaa}

			\draw[rotate around={90:(-3.9,-1.2)}, color=purple, fill=purple, ultra thick, opacity=0.3] (-3.9,-1.2) ellipse (0.35cm and 0.35cm); % 1 
			\draw[rotate around={90:(-2.5,-1.2)}, color=purple, fill=purple, ultra thick, opacity=0.3] (-2.5,-1.2) ellipse (0.35cm and 0.35cm); % 2 	
			\draw[rotate around={90:(-0.9,-1.2)}, color=purple, fill=purple, ultra thick, opacity=0.3] (-0.9,-1.2) ellipse (0.35cm and 0.35cm); % 3 
			\draw[rotate around={90:(0.5,-1.2)}, color=purple, fill=purple, ultra thick, opacity=0.3] (0.5,-1.2) ellipse (0.35cm and 0.35cm); % 4 
			
			\draw[rotate around={0:(-3.2,-0.2)}, color=blue, fill=blue, ultra thick, opacity=0.3] (-3.2,-0.2) ellipse (0.5cm and 0.35cm);% 5
			\draw[rotate around={0:(-0.2,-0.2)}, color=blue, fill=blue, ultra thick, opacity=0.3] (-0.2,-0.2) ellipse (0.5cm and 0.35cm);% 6 
			
			\draw[rotate around={0:(-3.2,0.8)}, color=green, fill=green, ultra thick, opacity=0.3] (-3.2,0.8) ellipse (0.5cm and 0.35cm); % 7
			\draw[rotate around={0:(-0.2,0.8)}, color=green, fill=green, ultra thick, opacity=0.3] (-0.2,0.8) ellipse (0.5cm and 0.35cm); % 9 
			
			\draw[rotate around={0:(-1.7,2)}, color=cyan, fill=cyan, ultra thick, opacity=0.3] (-1.7,2) ellipse (1cm and 0.35cm); % 10 	
			
			\node (1) at (-3.9,-1.2) {$1$};
			\node (2) at (-2.5,-1.2) {$2$};
			\node (3) at (-0.9,-1.2) {$5$};
			\node (4) at (0.5,-1.2) {$6$};
			
			\node (5) at (-3.2,-0.2) {$123$};
			\node (6) at (-0.2,-0.2) {$456$};
			\node (7) at (-3.2,0.8) {$1234$};
			\node (9) at (-0.2,0.8) {$3456$};
			\node (10) at (-1.7,2) {$12345678$};
			
			\node (11) at (-1.7,-2.2) {(b) $Q$};
			\node (11) at (1.2,-0.2) {{\boldmath$\leftarrow$}};
			
			\draw[->] (1)--(5) node[midway, above]{};
			\draw[->] (2)--(5) node[midway, above]{};
			\draw[->] (3)--(6) node[midway, above]{};
			\draw[->] (4)--(6) node[midway, above]{};
			\draw[->] (5)--(7) node[midway, above]{};				
			\draw[->] (6)--(9) node[midway, above]{};
			\draw[->] (7)--(10) node[midway, above]{};
			\draw[->] (9)--(10) node[midway, above]{};
			
			\label{bbbb}	
		\end{tikzpicture}%\medskip
		
		\begin{tikzpicture}[xscale=0.6, yscale=0.5]

			\tikzstyle{every node}=[draw, shape=circle, inner sep=1.4pt, fill=black]
			\tikzstyle{every path}=[color=black]
			
			\draw[fill=gray!50!white] (0,0.5) ellipse (25mm and 20mm);
			\node (0) at (-1,-1) {};
			\node (00) at (1,-1) {};
			\node (1) at (-2,0) {};
			\node (2) at (-2,1) {};
			\node (3) at (-1,2) {};
			\node (4) at (1,2) {};
			\node (5) at (2,1) {};
			\node (6) at (2,0) {};
			\draw (0)--(00)--(1)--(2)--(3)--(4)--(5)--(6)--(0)--(1)--(3)--(5)--(0)--(2)--(5)--(00)--(4);
			\draw (00)--(2)--(4)--(6)--(1)--(4)--(0)--(3);
			\draw (00)--(3)--(6)--(00); \draw (2)--(6); \draw (1)--(5);
			
			\node[draw=none,fill=none] (111) at (0,-3.3) {(a') $G$};
			\label{a'}	
			\node[draw=none,fill=none] (11) at (10,-3.3) {(b') $H$};
			\label{b'}	
			
			\tikzstyle{every node}=[draw, color=black, shape=circle, inner sep=2pt, opacity=0.5]
			\tikzstyle{every path}=[color=black]
			\node[fill=purple] (n1) at (-3,-0.5) {};
			\node[fill=purple] (n2) at (-3.5,0.5) {};
			\draw (n1)--(1); \draw (n2)--(2);
			\node[fill=blue] (n123) at (-4,2.5) {};
			\draw (2)--(n123)--(1)--(n123)--(3);
			\node[fill=green] (n1234) at (-3,3.5) {};
			\draw (2)--(n1234)--(1)--(n1234)--(3)--(n1234)--(4);
			\node[fill=purple] (n6) at (3,-0.5) {};
			\node[fill=purple] (n5) at (3.5,0.5) {};
			\draw (n5)--(5); \draw (n6)--(6);
			\node[fill=blue, thick] (n356) at (4,2.5) {};
			\draw[thick] (3)--(n356)--(5)--(n356)--(6);
			\node[fill=green] (n3456) at (3,3.5) {};
			\draw (5)--(n3456)--(6)--(n3456)--(3)--(n3456)--(4);
			
			\node[fill=cyan] (nall) at (0,-2.5) {};
			\draw (2)--(nall)--(1)--(nall)--(3)--(nall)--(4)--(nall)--(5)--(nall)--(6)--(nall)--(0)--(nall)--(00);

			\tikzstyle{every node}=[draw, shape=circle, inner sep=1.4pt, fill=black]
			\tikzstyle{every path}=[color=black]
			
			\draw[fill=gray!50!white] (10,0.5) ellipse (25mm and 20mm);
			\node (0) at (9,-1) {};
			\node (00) at (11,-1) {};
			\node (1) at (8,0) {};
			\node (2) at (8,1) {};
			\node (3) at (9,2) {};
			\node (4) at (11,2) {};
			\node (5) at (12,1) {};
			\node (6) at (12,0) {};
			\draw (0)--(00)--(1)--(2)--(3)--(4)--(5)--(6)--(0)--(1)--(3)--(5)--(0)--(2)--(5)--(00)--(4);
			\draw (00)--(2)--(4)--(6)--(1)--(4)--(0)--(3);
			\draw (00)--(3)--(6)--(00); \draw (2)--(6); \draw (1)--(5);
			
			\tikzstyle{every node}=[draw, color=black, shape=circle, inner sep=2pt, opacity=0.5]
			\tikzstyle{every path}=[color=black]
			\node[fill=purple] (n1) at (7,-0.5) {};
			\node[fill=purple] (n2) at (6.5,0.5) {};
			\draw (n1)--(1); \draw (n2)--(2);
			\node[fill=blue] (n123) at (6,2.5) {};
			\draw (2)--(n123)--(1)--(n123)--(3);
			\node[fill=green] (n1234) at (7,3.5) {};
			\draw (2)--(n1234)--(1)--(n1234)--(3)--(n1234)--(4);
			\node[fill=purple] (n6) at (13,-0.5) {};
			\node[fill=purple] (n5) at (13.5,0.5) {};
			\draw (n5)--(5); \draw (n6)--(6);
			\node[fill=blue, thick] (n456) at (14,2.5) {};
			\draw[thick] (4)--(n456)--(5)--(n456)--(6);
			\node[fill=green] (n3456) at (13,3.5) {};
			\draw (5)--(n3456)--(6)--(n3456)--(3)--(n3456)--(4);
			
			\node[fill=cyan] (nall) at (10,-2.5) {};
			\draw (2)--(nall)--(1)--(nall)--(3)--(nall)--(4)--(nall)--(5)--(nall)--(6)--(nall)--(0)--(nall)--(00);
			
		\end{tikzpicture}\qquad
		\vspace*{-2ex}%
		
		\caption{Two isomorphic colored posets (a) $P$ and (b) $Q$ obtained from non-isomorphic $S_d$-graphs (a') $G$ and (b') $H$ with the central cliques shaded gray.}
		\label{fig:sharedelements}	
	\end{figure}
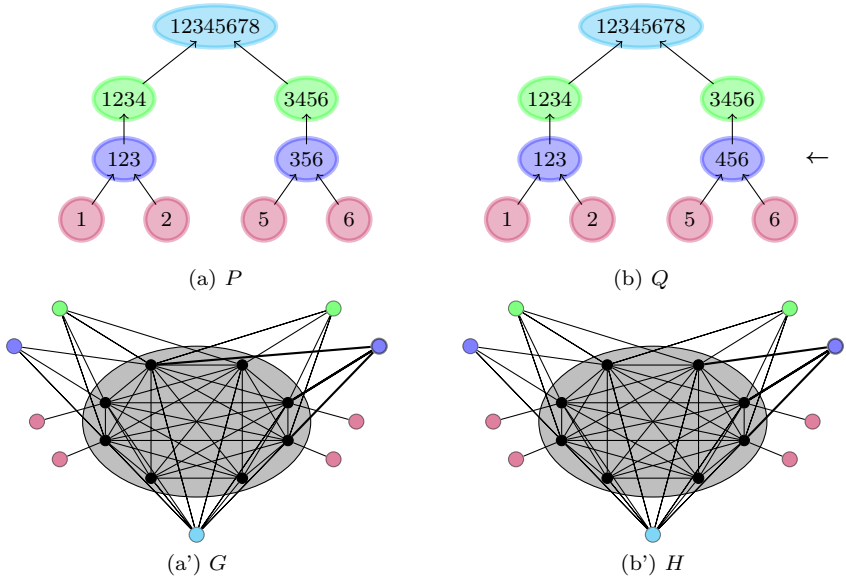

	We next move to step 3 of Procedure~\ref{proc:isomore}, and show, in Lemma~\ref{lem:preciseVenn},
	that it suffices to ensure that problems with cardinalities of the
	intersections of attachment sets like the one in the example do not happen, to claim that $G$ and~$H$ are indeed isomorphic.

	\subsection{Checking the consistency on the central cliques}
	
	Recall the automorphism group $\Gamma$ of $R$ from step~2 of Procedure~\ref{proc:isomore}.
	We say that an automorphism $\varrho\in\Gamma$ is {\em consistent on the central cliques~$C\cup D$}
	if $\varrho$ satisfies the condition in step 3 of Procedure~\ref{proc:isomore}, that is, if
	\begin{itemize}
		\item[*] there exists a permutation $f_{\varrho}$ of $C\cup D$ such
		that, for every component $Z\in\mathcal{X}\cup\mathcal{Y}$,
		there is a respectful isomorphism from $K(Z)$ to $K(\varrho(Z))$
		whose restriction to $C\cup D$ equals the respective restriction of~$f_{\varrho}$.
	\end{itemize}
	While we explicitly processed all bijections between the central cliques $C$ and $D$ 
	in the simple approach given in Section~\ref{easysection}, 
	we consider no explicit permutations on $C\cup D$ in this general case.
	Instead, we will indirectly check for an existence of a permutation $f_{\varrho}$
	on $C\cup D$ witnessing consistency of $\varrho\in\Gamma$ on the central~cliques.
	
	Observe that, for a component $Z\in\mathcal{X}\cup\mathcal{Y}$ (of $K-(C\cup D)=(G-C)\uplus(H-D)$), 
	the $a\geq1$ distinct neighborhoods (attachment sets) of vertices of $Z$ in $C\cup D$ 
	form a sequence $N_1(Z),\ldots,N_a(Z)\subseteq C\cup D$ ordered by the strict inclusion
	$N_{C\cup D}^{L}(Z)=N_1(Z)\subsetneq N_2(Z)\subsetneq\cdots\subsetneq N_a(Z)=N_{C\cup D}^{U}(Z)$,
	and we denote this whole family by $\mathcal{N}_{C\cup D}(Z):=\{N_1(Z),\ldots,N_a(Z)\}$.
	
	\vskip1pt
	We call the {\em multiset} of sets
	$\mathcal{U}:=\biguplus_{Z\in\mathcal{X}\cup\mathcal{Y}}\, \mathcal{N}_{C\cup D}(Z)$
	the {\em attachment collection} of $\mathcal{X}\cup\mathcal{Y}$ in $C\cup D$ of our graph~$K$
	(\,$\mathcal{U}$ is a multiset since the same attachment set may occur several times in
	$\mathcal{U}$ if the occurrences come from distinct bridges of $C\cup D$
	which are incomparable in~$R$). Recall the notation $K(Z)$ from Procedure~\ref{proc:isomore}.
	If an automorphism $\varrho$ of~$R$ maps $Z$ to $Z'=\varrho(Z)$, then,
	in particular, $K(Z)$ is respectfully isomorphic to $K(Z')$.
	While there may exist different respectful isomorphisms from $K(Z)$ to $K(Z')$,
	they all define (because of the strict inclusion order) the same unique mapping
	of the attachment sets from~$\mathcal{N}_{C\cup D}(Z)$ to $\mathcal{N}_{C\cup D}(Z')$.
	So, the automorphism $\varrho$ of $R$ induces a unique corresponding
	permutation on~$\mathcal{U}$, denoted here by~$\tilde\varrho$.
	
	For a set family $\mathcal{U}$, we call
	a {\em cardinality Venn diagram} of $\mathcal{U}$ the vector
	$\big(\ell_{\mathcal{U},\mathcal{U}_1}: \emptyset\not=\mathcal{U}_1\subseteq\mathcal{U}\big)$
	such that $\ell_{\mathcal{U},\mathcal{U}_1}:=|L_{\mathcal{U},\mathcal{U}_1}|$ where
	\vspace{1pt}%
	$L_{\mathcal{U},\mathcal{U}_1}=\bigcap_{A\in\mathcal{U}_1}\!A
	\setminus \bigcup_{B\in{\mathcal{U}\setminus\mathcal{U}_1}}\!B$.
	That is, we record the cardinality of every internal cell of the Venn diagram
	of~$\mathcal{U}$.
	Let $\tilde\varrho(\mathcal{U}_1)=\{\tilde\varrho(A):A\in\mathcal{U}_1\}$
	for $\mathcal{U}_1\subseteq\mathcal{U}$.

	\begin{figure}[tbp]
		\centering
		\captionsetup[subfigure]{position=b,justification=centering}
		\begin{subfigure}[t]{0.48\linewidth}
			\centering
			\begin{tikzpicture}[scale=0.8]
				
				\filldraw [fill=cyan, draw=cyan, ultra thick, opacity=0.3] (-4.6,2.35) rectangle (1.6,-1.2); %C1
				
				\draw[rotate around={0:(-2.6,0.4)}, color=green, fill=green, ultra thick, opacity=0.2] (-2.6,0.4) ellipse (1.8cm and 1.45cm); % 1234
				\draw[rotate around={0:(-0.4,0.4)}, color=green, fill=green, ultra thick, opacity=0.2] (-0.4,0.4) ellipse (1.8cm and 1.45cm); % 3456
				
				\draw[rotate around={0:(-2.6,0)}, color=blue, fill=blue, ultra thick, opacity=0.2] (-2.6,0) ellipse (1.6cm and 0.7cm); % 123
				\draw[rotate around={0:(-0.4,0)}, color=blue, fill=blue, ultra thick, opacity=0.2] (-0.4,0) ellipse (1.6cm and 0.7cm); % 356
				
				\draw[rotate around={90:(-3.5,0)}, color=purple, fill=purple, ultra thick, opacity=0.2] (-3.5,0) ellipse (0.35cm and 0.35cm); % 1	
				
				\draw[rotate around={90:(-2.5,0)}, color=purple, fill=purple, ultra thick, opacity=0.2] (-2.5,0) ellipse (0.35cm and 0.35cm); % 2	
				
				\draw[rotate around={90:(-0.5,0)}, color=purple, fill=purple, ultra thick, opacity=0.2] (-0.5,0) ellipse (0.35cm and 0.35cm); % 5		
				
				\draw[rotate around={90:(0.5,0)}, color=purple, fill=purple, ultra thick, opacity=0.2] (0.5,0) ellipse (0.35cm and 0.35cm); % 6
				
				\node (1) at (-3.5,0) {$1$};
				\node (2) at (-2.5,0) {$2$};
				\node (3) at (-1.5,0) {$3$}; 
				\node (4) at (-1.5,1) {$4$};
				\node (5) at (-0.5,0) {$5$}; 
				\node (6) at (0.5,0) {$6$};			
				\node (8) at (-2.5,2.1) {$7$};
				\node (9) at (-0.5,2.1) {$8$}; 	
				
				\node (11) at (-1.5,-1.8) {(a)};
				
			\end{tikzpicture}	
			\label{fig:a}
		\end{subfigure}
		~
		\begin{subfigure}[t]{0.48\linewidth}
			\centering
			\begin{tikzpicture}[scale=0.8]
				
				\filldraw [fill=cyan, draw=cyan, ultra thick, opacity=0.3] (2.7,2.35) rectangle (8.9,-1.2);
				
				\draw[rotate around={0:(4.8,0.3)}, color=green, fill=green, ultra thick, opacity=0.2] (4.8,0.3) ellipse (2.1cm and 1.2cm); % 1234
				\draw[rotate around={0:(6.8,0.3)}, color=green, fill=green, ultra thick, opacity=0.2] (6.8,0.3) ellipse (2.1cm and 1.2cm); % 3456
				
				\draw[rotate around={0:(4.3,0.3)}, color=blue, fill=blue, ultra thick, opacity=0.2] (4.3,0.3) ellipse (1.45cm and 0.8cm); % 123
				\draw[rotate around={0:(7.3,0.3)}, color=blue, fill=blue, ultra thick, opacity=0.2] (7.3,0.3) ellipse (1.45cm and 0.8cm); % 456
				
				\draw[rotate around={90:(3.3,0.3)}, color=purple, fill=purple, ultra thick, opacity=0.2] (3.3,0.3) ellipse (0.35cm and 0.35cm); % 1	
				
				\draw[rotate around={90:(4.3,0.3)}, color=purple, fill=purple, ultra thick, opacity=0.2] (4.3,0.3) ellipse (0.35cm and 0.35cm); % 2
				
				\draw[rotate around={90:(7.3,0.3)}, color=purple, fill=purple, ultra thick, opacity=0.2] (7.3,0.3) ellipse (0.35cm and 0.35cm); % 5			
				
				\draw[rotate around={90:(8.3,0.3)}, color=purple, fill=purple, ultra thick, opacity=0.2] (8.3,0.3) ellipse (0.35cm and 0.35cm); % 6	
				
				\node (1) at (3.3,0.3) {$1$}; 
				\node (2) at (4.3,0.3) {$2$}; 
				\node (3) at (5.3,0.3) {$3$}; 	
				\node (4) at (6.3,0.3) {$4$}; 
				\node (5) at (7.3,0.3) {$5$}; 
				\node (6) at (8.3,0.3) {$6$}; 		
				\node (8) at (4.8,2.1) {$7$}; 		
				\node (9) at (6.8,2.1) {$8$}; 
				
				\node (11) at (5.8,-1.8) {(b)};
				
			\end{tikzpicture}
			\label{fig:b}
		\end{subfigure}
		
		\begin{subfigure}[t]{0.48\linewidth}
			\centering
			\begin{tikzpicture}[scale=0.8]
				
				\filldraw [fill=cyan, draw=cyan, ultra thick, opacity=0.3] (-4.6,2.35) rectangle (1.6,-1.2); %C1
				
				\draw[rotate around={0:(-2.6,0.4)}, color=green, fill=green, ultra thick, opacity=0.2] (-2.6,0.4) ellipse (1.8cm and 1.45cm); % 1234
				\draw[rotate around={0:(-0.4,0.4)}, color=green, fill=green, ultra thick, opacity=0.2] (-0.4,0.4) ellipse (1.8cm and 1.45cm); % 3456
				
				\draw[rotate around={0:(-2.6,0)}, color=blue, fill=blue, ultra thick, opacity=0.2] (-2.6,0) ellipse (1.6cm and 0.7cm); % 123
				\draw[rotate around={0:(-0.4,0)}, color=blue, fill=blue, ultra thick, opacity=0.2] (-0.4,0) ellipse (1.6cm and 0.7cm); % 356
				
				\draw[rotate around={90:(-3.5,0)}, color=purple, fill=purple, ultra thick, opacity=0.2] (-3.5,0) ellipse (0.35cm and 0.35cm); % 1	
				
				\draw[rotate around={90:(-2.5,0)}, color=purple, fill=purple, ultra thick, opacity=0.2] (-2.5,0) ellipse (0.35cm and 0.35cm); % 2	
				
				\draw[rotate around={90:(-0.5,0)}, color=purple, fill=purple, ultra thick, opacity=0.2] (-0.5,0) ellipse (0.35cm and 0.35cm); % 5		
				
				\draw[rotate around={90:(0.5,0)}, color=purple, fill=purple, ultra thick, opacity=0.2] (0.5,0) ellipse (0.35cm and 0.35cm); % 6
				
				\boldmath
				\node (1) at (-3.5,0) {$1$};
				\node (2) at (-2.5,0) {$1$};
				\node (3) at (-1.5,0) {$1$}; 
				\node (4) at (-1.5,1) {$1$};
				\node (5) at (-0.5,0) {$1$}; 
				\node (6) at (0.5,0) {$1$};			
				\node (8) at (-1.5,2.1) {$2$};	
				
				\node (11) at (-1.5,-1.8) {(a')};
				
			\end{tikzpicture}
			\label{fig:a'}
		\end{subfigure}
		~
		\begin{subfigure}[t]{0.48\linewidth}
			\centering
			\begin{tikzpicture}[scale=0.8]
				
				\filldraw [fill=cyan, draw=cyan, ultra thick, opacity=0.3] (2.7,2.35) rectangle (8.9,-1.2);
				
				\draw[rotate around={0:(4.8,0.3)}, color=green, fill=green, ultra thick, opacity=0.2] (4.8,0.3) ellipse (2.1cm and 1.2cm); % 1234
				\draw[rotate around={0:(6.8,0.3)}, color=green, fill=green, ultra thick, opacity=0.2] (6.8,0.3) ellipse (2.1cm and 1.2cm); % 3456
				
				\draw[rotate around={0:(4.3,0.3)}, color=blue, fill=blue, ultra thick, opacity=0.2] (4.3,0.3) ellipse (1.45cm and 0.8cm); % 123
				\draw[rotate around={0:(7.3,0.3)}, color=blue, fill=blue, ultra thick, opacity=0.2] (7.3,0.3) ellipse (1.45cm and 0.8cm); % 456
				
				\draw[rotate around={90:(3.3,0.3)}, color=purple, fill=purple, ultra thick, opacity=0.2] (3.3,0.3) ellipse (0.35cm and 0.35cm); % 1	
				
				\draw[rotate around={90:(4.3,0.3)}, color=purple, fill=purple, ultra thick, opacity=0.2] (4.3,0.3) ellipse (0.35cm and 0.35cm); % 2
				
				\draw[rotate around={90:(7.3,0.3)}, color=purple, fill=purple, ultra thick, opacity=0.2] (7.3,0.3) ellipse (0.35cm and 0.35cm); % 5			
				
				\draw[rotate around={90:(8.3,0.3)}, color=purple, fill=purple, ultra thick, opacity=0.2] (8.3,0.3) ellipse (0.35cm and 0.35cm); % 6	
				
				\boldmath
				\node (1) at (3.3,0.3) {$1$}; 
				\node (2) at (4.3,0.3) {$1$}; 
				\node (3) at (5.3,0.3) {$1$}; 	
				\node (4) at (6.3,0.3) {$1$}; 
				\node (5) at (7.3,0.3) {$1$}; 
				\node (6) at (8.3,0.3) {$1$}; 		
				\node (8) at (5.8,2.1) {$2$}; 		
				
				\node (11) at (5.8,-1.8) {(b')};
				
			\end{tikzpicture}
			\label{fig:b'}
		\end{subfigure}
		
		\caption{The Venn diagrams (a) and (b) and the corresponding cardinality Venn diagrams (a') and (b') for the colored posets $P$ and $Q$ given in Figure~\ref{fig:sharedelements}. }
		\label{fig:cardinalityvenndiagrams}	
	\end{figure}
	
	It is reasonably easy to see that an automorphism $\varrho$ of~$R$ is consistent on the central
	cliques $C\cup D$ of $K$, if and only if the corresponding permutation $\tilde\varrho$
	on the attachment collection $\mathcal{U}$ of~$\mathcal{X}\cup\mathcal{Y}$
	preserves the values of all cells of the cardinality Venn diagram of $\mathcal{U}$.
	This is precisely formulated as follows:
	
	\begin{lem}\label{lem:preciseVenn}
		Let $K=G\uplus H$ and~$C,D$, the sets $\mathcal{X}$ and $\mathcal{Y}$,
		and the posets $P,Q$ and~$R$ (on the ground set $\mathcal{X}\cup\mathcal{Y}$) 
		be as in Procedure~\ref{proc:isomore}.
		Let $\mathcal{U}$ be the {attachment collection}
		of~$\mathcal{X}\cup\mathcal{Y}$ in~$C\cup D$, and
		assume an automorphism $\varrho$ of $R$ and the corresponding permutation
		$\tilde\varrho$ of~$\mathcal{U}$.
		
		There is an automorphism $f$ of the graph $K$ such that $f(C\cup D)=C\cup D$
		and, for every component $Z\in\mathcal{X}\cup\mathcal{Y}$,~ 
		$f$ maps $V(Z)$ to $V(\varrho(Z))$,
		if and only if the cardinality Venn diagrams of $\mathcal{U}$ and of 
		$\tilde\varrho(\mathcal{U})$ are the same, meaning that
		$\ell_{\mathcal{U},\mathcal{U}_1}=
		\ell_{\mathcal{U},\tilde\varrho(\mathcal{U}_1)}$
		for all~$\emptyset\not=\mathcal{U}_1\subseteq\mathcal{U}$.
	\end{lem}
	
	\begin{proof}
		$\Rightarrow$ 
		Suppose that there exists an automorphism $f$ of $K$ such that
		$f(C\cup D)=C\cup D$ and, for every component
		$Z\in\mathcal{X}\cup\mathcal{Y}$, $f$ maps the vertices of $Z$ to the
		vertices of~$\varrho(Z)$.  
		Then, each pair $K(Z)$ and $K(\varrho(Z))$ are isomorphic interval graphs, 
		and if there are $a$ distinct neighborhoods $N_1(Z),\dots,N_a(Z)$ of vertices 
		of $Z$ in $C \cup D$, then there are $a$
		neighborhoods $N_1(\varrho(Z)),\dots, N_a(\varrho(Z))$ of vertices of
		$\varrho(Z)$ in $C \cup D$ determined by the restriction of $f$ to~$C\cup D$.
		Since $f$ is an automorphism of $K$, indeed,
		$\tilde\varrho(N_i(Z)) =N_i(\varrho(Z))$ for $1 \leq i \leq a$.  
		Moreover, by our assumption, if $v\in N_i(Z)$, then 
		$f(v)\in N_i(\varrho(Z))$, % ($=\tilde\varrho(N_i(Z))$),
		and vice versa.
		Now consider arbitrary $\emptyset\not=\mathcal{U}_1\subseteq\mathcal{U}$.
		By the previous; if
		$$v\in \bigcap_{A\in\mathcal{U}_1}\!A \setminus
		\bigcup_{B\in{\mathcal{U}\setminus\mathcal{U}_1}}\!B
		,$$
		then (since $\tilde\varrho$ is a permutation of~$\mathcal{U}$)
		$$f(v)\in \bigcap_{A\in\mathcal{U}_1}\!\tilde\varrho(A) \setminus
		\bigcup_{B\in{\mathcal{U}\setminus\mathcal{U}_1}}\!\tilde\varrho(B)
		~= \bigcap_{A\in\tilde\varrho(\mathcal{U}_1)}\!A \setminus
		\bigcup_{B\in{\mathcal{U}\setminus\tilde\varrho(\mathcal{U}_1)}}\!B
		,$$
		and vice versa.
		Consequently, $\ell_{\mathcal{U},\mathcal{U}_1}=
		\ell_{\mathcal{U},\tilde\varrho(\mathcal{U}_1)}$
		for all~$\emptyset\not=\mathcal{U}_1\subseteq\mathcal{U}$.
		
		\smallskip
		$\Leftarrow$ We start with a simple claim:
		if, for some $Z\in\mathcal{X}\cup\mathcal{Y}$,
		$f$ is a permutation of $C\cup D$ which set-wise stabilizes all
		sets in $\mathcal{N}_{C\cup D}(Z)$, then $f$ extended with the identity on
		$Z$ is a respectful automorphism of~$K(Z)$.
		This is trivially true for edges of $K(Z)$ having both ends either in $Z$ or in~$C\cup D$.
		For $v\in C\cup D$ and~$w\in Z$, we have $vw\in E(K(Z))$ $\iff$
		$f(v)w\in E(K(Z))$ since otherwise the neighborhood set of $w$ in $C\cup D$
		would not be stabilized by~$f$.
		
		Now suppose that $\ell_{\mathcal{U},\mathcal{U}_1}=\ell_{\mathcal{U},\tilde\varrho(\mathcal{U}_1)}$
		holds for our automorphism $\varrho$ of $R$ and all~$\emptyset\not=\mathcal{U}_1\subseteq\mathcal{U}$.  
		Then, in particular, for every such $\mathcal{U}_1$ there exists
		$$\mbox{ a bijection from }
		\bigcap_{A\in\mathcal{U}_1}\!A \setminus
		\bigcup_{B\in{\mathcal{U}\setminus\mathcal{U}_1}}\!B
		\mbox{ ~to } \bigcap_{A\in\tilde\varrho(\mathcal{U}_1)}\!A \setminus
		\bigcup_{B\in{\mathcal{U}\setminus\tilde\varrho(\mathcal{U}_1)}}\!B
		.$$
		The composition of these bijections results in a permutation $f_0$ of~$C\cup D$.
		For every $Z\in\mathcal{X}\cup\mathcal{Y}$, since $\varrho$ is an
		automorphism of~$R$, there is a respectful isomorphism $f_Z$ from $K(Z)$ to $K(\varrho(Z))$.
		Let $f_Z'$ be the restriction of $f_Z$ to $C\cup D$
		and let $f_0':={f_Z'}^{-1}\circ f_0$ extended with the identity map on~$Z$.
		Then $f_0'$ set-wise stabilizes all sets in $\mathcal{N}_{C\cup D}(Z)$ by
		the definition of $\tilde\varrho$ and $f_0$.
		Consequently, $f_0'$ is a respectful automorphism of~$K(Z)$ and so 
		$f_Z^o:=f_Z\circ f_0'$ is a respectful isomorphism from $K(Z)$ to $K(\varrho(Z))$
		which coincides with $f_0$ on~$C\cup D$.
		
		Finally, since the members of $\mathcal{X}\cup\mathcal{Y}$ are pairwise disjoint
		sets, the composition of $f_Z^o$ over all $Z\in\mathcal{X}\cup\mathcal{Y}$
		is well-defined and it is hence an automorphism of the graph $K$ satisfying the desired properties.
		\qed\end{proof}

	At first sight, the condition of Lemma~\ref{lem:preciseVenn} may not seem
	efficient since $\mathcal{U}$ has up to $2n$ attachment sets, 
	and so an exponential number of Venn diagram cells.
	Though, only at most $2n$ of the cells may be nonempty since the ground set of
	$\mathcal{U}$ is of cardinality $|C\cup D|\leq2n$,
	and so we can handle the situation as follows:
	
	\begin{lem}\label{lem:testconsist}
		For any $\mathcal{U'}\subseteq\mathcal{U}$ such that
		$\tilde\varrho(\mathcal{U'})=\mathcal{U'}$,
		one can in $\mathcal{O}(n^2)$ time test whether the equalities
		$\ell_{\mathcal{U}'\!,\,\mathcal{U}_1}=   
		\ell_{\mathcal{U}',\tilde\varrho(\mathcal{U}_1)}$
		hold for all~$\emptyset\not=\mathcal{U}_1\subseteq\mathcal{U}'$.
	\end{lem}
	
	\begin{proof}
		We loop through all vertices $w$ of $C\cup D$, and for each $w$ we
		record in $\mathcal{O}(n)$ time to which of the sets in $\mathcal{U'}$ this $w$ belongs to.
		Summing the obtained records at the~end precisely gives the $\mathcal{O}(n)$ nonzero values
		$\ell_{\mathcal{U}'\!,\,\mathcal{U}_1}$ over $\emptyset\not=\mathcal{U}_1\subseteq\mathcal{U}'$.
		
		We analogously compute the $\mathcal{O}(n)$ nonzero values       
		$\ell_{\mathcal{U}',\tilde\varrho(\mathcal{U}_1)}$ over $\emptyset\not=\mathcal{U}_1\subseteq\mathcal{U}'$,
		and then compare the two sets of values with respect to
		each $\mathcal{U}_1$ and matching $\tilde\varrho(\mathcal{U}_1)$.\qed
	\end{proof}

	\subsection{Computing the subgroup of consistent poset automorphisms}
	
	Knowing how to efficiently test whether an automorphism $\varrho\in\Gamma$
	in step 3 of Procedure~\ref{proc:isomore} is consistent
	(Lemmas~\ref{lem:preciseVenn} and~\ref{lem:testconsist}),
	we would like to finish a computation of the subgroup $\Gamma'\subseteq\Gamma$.
	This, however, cannot be done directly by processing all members of
	the group $\Gamma$ which can be exponentially large.
	Instead, inspired by famous Babai's ``tower-of-groups'' procedure
	(cf.~\cite{babai-bdcm} and Theorem~\ref{thm:bdcm}),
	we iteratively compute a chain of subgroups
	$\Gamma=\Gamma_0\supseteq\Gamma_1\supseteq\ldots\supseteq\Gamma_h=\Gamma'$ leading to the result.
	Here, by ``computing a~group'' we mean to output a set of its generators. 
	
	We will show that this step can also be executed in \textbf{FPT}-time. There are two important ingredients making this computation work.
	First, we look for a manageable combinatorially defined ``gradual refinement'' of the
	condition tested by Lemma~\ref{lem:preciseVenn}.
	Our intention is to define $\Gamma_i$ as the subgroup of $\Gamma$ respecting
	the $i$-th step of this refinement. % of Lemma~\ref{lem:testconsist}.
	By manageable we mean that the ratio of orders (sizes) of consequent groups
	$\Gamma_i$ and $\Gamma_{i+1}$ in the chain is always bounded and, at the same time, 
	that the number of refinement steps ($h$) is not too big.
	Second, having such manageable refinement steps, we then stepwise apply another
	classical result (as illustrated below):
	
	\begin{restatable}{theorem}{thmsubgroup}{\bf(Furst, Hopcroft and Luks
			\cite[Cor.~1]{furst})}\label{thm:furstgen}
		Let $\Pi$ be a permutation group given by its generators, 
		and $\Pi_1$ be any subgroup of $\Pi$ such that one can test in
		polynomial time whether $\pi\in\Pi_1$ for any $\pi\in\Pi$ (membership test).
		If the ratio $|\Pi|/|\Pi_1|$ is bounded by a function of a parameter $d$,
		then a set of generators of $\Pi_1$ can be computed in
		\textbf{FPT}-time (with respect to~$d$).
	\end{restatable}

	To illustrate the typical use and the strength of Theorem~\ref{thm:furstgen}, we show an outline of how 
	it can be used to design an algorithm (known as Babai's tower-of-groups
	procedure~\cite{babai-bdcm}) proving Theorem~\ref{thm:bdcm}:
	\begin{enumerate}
		\item Let $G$ be a graph on $n$ vertices, and $V(G)=V_1\cup\ldots\cup V_k$
		be a partition of its vertex set into $k$ color classes such that $|V_i|\leq d$.
		Let $\Pi_0$ be the group of all permutations $\pi$ on~$V(G)$ which satisfy
		$\pi(V_j)=V_j$ for all $1\leq j\leq k$
		(i.e., $\Pi_0$ is a product of the symmetric groups of the color classes).
		\item 
		For $i=1,2,\ldots,{k\choose2}$,
		let $(a,b)$ be the $i$-th pair in the following sequence of pairs:
		$(1,2),(1,3),\ldots,(1,k),(2,3),(2,4),\ldots,(2,k),(3,4),\ldots,(k-1,k)$
		(the order of which is not really important).
		Denote by $E_i$ the set of edges of the induced subgraph $G[V_a\cup V_b]$,
		and by $\Pi_i$ the subgroup of $\Pi_{i-1}$
		consisting of all permutations $\pi\in\Pi_{i-1}$ such that
		$\pi$ induces an automorphism of the graph with the vertex set $V(G)$ and
		the edge set $E_1\cup E_2\cup\ldots\cup E_i$.
		Then $|\Pi_{i-1}|/|\Pi_i|\leq |V_a|!\cdot|V_b|!\leq(d!)^2$, and so we may use
		Theorem~\ref{thm:furstgen} to compute $\Pi_i$ from $\Pi_{i-1}$.
		\item 
		Finally, $\Pi_{k\choose2}$ is the color-preserving automorphism group of~$G$, as desired.
	\end{enumerate}
	
	In a nutshell, the outlined procedure stepwise removes those permutations of
	$\Pi_0$ which violate the edge subsets between the parts $V_a$ and~$V_b$,
	ranging over all pairs $i=(a,b)$.
	To detail an analogous procedure in our case, we need to
	closely analyze what happens if an automorphism $\varrho\in\Gamma$ does
	not pass the cardinality Venn diagram test described in Lemma~\ref{lem:preciseVenn}.
	This is not nearly as simple as in case of the procedure for Theorem~\ref{thm:bdcm},
	since the cardinality Venn diagram involves attachment sets from all levels of the poset at once.
	Fortunately, thanks to considering posets of bounded width, we can prove
	that every violation of the consistency check from Lemma~\ref{lem:preciseVenn} is witnessed 
	by a subcollection of at most $d$ sets of $\mathcal{U}$,
	i.e., such failure involves only at most $d$ levels of the posets
	which is manageable.
	
	\begin{lem}\label{lem:Venngood}
		Let $\mathcal{U}$ be a set family and $\tilde\varrho$ be a permutation of $\mathcal{U}$.
		If there exists $\mathcal{U}_1$ such that 
		$\emptyset\not=\mathcal{U}_1\subseteq\mathcal{U}$ and
		$\ell_{\mathcal{U},\mathcal{U}_1}\not=\ell_{\mathcal{U},\tilde\varrho(\mathcal{U}_1)}$,
		then there exist $\mathcal{U}_2,\mathcal{U}_3\subseteq\mathcal{U}$ such that
		$|\mathcal{U}_2|\leq2$ or $\mathcal{U}_2$ is an antichain in the inclusion,
		$\emptyset\not=\mathcal{U}_3\subseteq\mathcal{U}_2$ and
		$\ell_{\mathcal{U}_2,\mathcal{U}_3}\not=
		\ell_{\tilde\varrho(\mathcal{U}_2),\tilde\varrho(\mathcal{U}_3)}$.
	\end{lem}
	In our case, $\mathcal{U}_2$ is an antichain of attachment sets of
	one of $G$ or~$H$, and hence $|\mathcal{U}_2|\leq d$ follows from the
	assumption of considering $S_d$-graphs.
	
	For simplicity, we say that $\mathcal{U}'\subseteq\mathcal{U}$ is 
	{\em Venn-good} if
	$\ell_{\mathcal{U}',\mathcal{U}_0} =
	\ell_{\tilde\varrho(\mathcal{U}'),\tilde\varrho(\mathcal{U}_0)}$
	holds true for all $\emptyset\not=\mathcal{U}_0\subseteq\mathcal{U}'$,
	and we call $\mathcal{U}_0$ a {\em witness} (of $\mathcal{U}'$ not being Venn-good) if 
	$\ell_{\mathcal{U}',\mathcal{U}_0} \not=
	\ell_{\tilde\varrho(\mathcal{U}'),\tilde\varrho(\mathcal{U}_0)}$.
	Recalling that the permutation $\tilde\varrho$ of $\mathcal{U}$ is determined by an
	automorphism $\varrho$ of our poset $R$, we also more precisely say that
	$\mathcal{U}'$ is Venn-good {\em for} this automorphism $\varrho$.
	
	\begin{proof}[of Lemma~\ref{lem:Venngood}]
		Notice that $\mathcal{U}$ itself is not Venn-good, and $\mathcal{U}_1$ is a witness.
		Choose $\mathcal{U}_2\subseteq\mathcal{U}$ such that $\mathcal{U}_2$ is
		{\em not} Venn-good and it is minimal such by inclusion,
		and assume (for a contradiction) that
		there are $A_1,A_2\in\mathcal{U}_2$ such that~$A_1\subseteq A_2$.
		If $\tilde\varrho(A_1)\not\subseteq\tilde\varrho(A_2)$, then already
		$\mathcal{U}_2:=\{A_1,A_2\}$ is not Venn-good (with a witness $\{A_1\}$), 
		and so let $\tilde\varrho(A_1)\subseteq\tilde\varrho(A_2)$.
		
		Let $\mathcal{U}_3$ be a witness of $\mathcal{U}_2$ not being Venn-good,
		and for $j=2,3$ denote: 
		$\mathcal{U}_j^0:=\mathcal{U}_j\setminus\{A_1,A_2\}$,
		$\mathcal{U}_j^1:=\big(\mathcal{U}_j\cup\{A_1\}\big)\setminus\{A_2\}$,
		$\mathcal{U}_j^2:=\big(\mathcal{U}_j\cup\{A_2\}\big)\setminus\{A_1\}$,
		$\mathcal{U}_j^3:=\mathcal{U}_j\cup\{A_1,A_2\}$.
		By our minimality assumption, all three subfamilies
		$\mathcal{U}_2^0$, $\mathcal{U}_2^1$ and $\mathcal{U}_2^2$ are Venn-good.
		We first easily derive
		\begin{eqnarray*}
			\ell_{\mathcal{U}_2,\mathcal{U}_3^0} = 
			\ell_{\mathcal{U}_2\setminus\{\!A_1\!\},\,\mathcal{U}_3^0} =
			\ell_{\mathcal{U}_2^2,\,\mathcal{U}_3^0} 
			&=&
			\ell_{\tilde\varrho(\mathcal{U}_2^2),\tilde\varrho(\mathcal{U}_3^0)} =
			\ell_{\tilde\varrho(\mathcal{U}_2),\tilde\varrho(\mathcal{U}_3^0)}
			\,,\\
			\ell_{\mathcal{U}_2,\mathcal{U}_3^1} = 0 &=& 0 =
			\ell_{\tilde\varrho(\mathcal{U}_2),\tilde\varrho(\mathcal{U}_3^1)}
			\,,\\
			\ell_{\mathcal{U}_2,\mathcal{U}_3^3} = 
			\ell_{\mathcal{U}_2\setminus\{\!A_2\!\},\,\mathcal{U}_3^3\setminus\{\!A_2\!\}} 
			=\ell_{\mathcal{U}_2^1,\,\mathcal{U}_3^1} 
			&=&
			\ell_{\tilde\varrho(\mathcal{U}_2^1),\tilde\varrho(\mathcal{U}_3^1)} =
			\ell_{\tilde\varrho(\mathcal{U}_2),\tilde\varrho(\mathcal{U}_3^3)}
			\,.\end{eqnarray*}
		Then, using trivial
		$\ell_{\mathcal{U}_2^0,\,\mathcal{U}_3^0}=
		\ell_{\mathcal{U}_2,\mathcal{U}_3^0}+\ell_{\mathcal{U}_2,\mathcal{U}_3^1}
		+\ell_{\mathcal{U}_2,\mathcal{U}_3^2}+\ell_{\mathcal{U}_2,\mathcal{U}_3^3}$
		and its counterpart under $\tilde\varrho$, we conclude
		\begin{eqnarray*}
			\ell_{\mathcal{U}_2,\mathcal{U}_3^2} &=&
			\ell_{\mathcal{U}_2^0,\,\mathcal{U}_3^0}
			-\ell_{\mathcal{U}_2,\mathcal{U}_3^0}-0
			-\ell_{\mathcal{U}_2,\mathcal{U}_3^3}
			\\	&=&\ell_{\tilde\varrho(\mathcal{U}_2^0),\tilde\varrho(\mathcal{U}_3^0)}
			-\ell_{\tilde\varrho(\mathcal{U}_2),\tilde\varrho(\mathcal{U}_3^0)}
			-\ell_{\tilde\varrho(\mathcal{U}_2),\tilde\varrho(\mathcal{U}_3^3)}
			=\>\ell_{\tilde\varrho(\mathcal{U}_2),\tilde\varrho(\mathcal{U}_3^2)}
			\,.\end{eqnarray*}
		However, $\mathcal{U}_3\in\{\mathcal{U}_3^0,
		\mathcal{U}_3^1,\mathcal{U}_3^2,\mathcal{U}_3^3\}$,
		and so one of the latter four equalities contradicts the assumption 
		that $\mathcal{U}_3$ witnessed $\mathcal{U}_2$ not being Venn-good.\qed
	\end{proof}
	
	\begin{cor}\label{cor:d-tuples}
		Let a poset $R$, its automorphism $\varrho$, attachment collection $\mathcal{U}$ 
		and permutation $\tilde\varrho$ of $\mathcal{U}$ be as in Lemma~\ref{lem:preciseVenn}.
		We have that $\mathcal{U}$ is Venn-good (wrt.~$\tilde\varrho$),
		if and only if every $\mathcal{U}'$ is Venn-good,
		where $\mathcal{U}'\subseteq\mathcal{U}$ is the subcollection
		of attachment sets of the union of some (any) $d$ levels of the poset~$R$.
		\qed\end{cor}
	
	Corollary~\ref{cor:d-tuples} shows a clear road to computing the subgroup
	$\Gamma'\subseteq\Gamma$ in step 3 of Proce\-dure~\ref{proc:isomore}.
	In every refinement step of a chain
	$\Gamma=\Gamma_0\supseteq\Gamma_1\supseteq\cdots\supseteq\Gamma_h=\Gamma'$,
	the following constraint is imposed: the subcollection of attachment sets of some
	$d$-tuple of levels of the poset~$R$ is Venn-good for every member of the next subgroup.
	All these steps are manageable; the ratio $|\Gamma_i|/|\Gamma_{i+1}|$ is
	bounded from above by the maximum number of subpermutations of
	$\Gamma_i$ on the respective $d$ levels (as proved below), which is $\leq(2d)!^{d}$.
	Since the height of $R$ is $\Theta(n)$, though, we cannot afford to check
	all $d$-tuples of levels this way.
	Fortunately, it is also not necessary by the following argument.
	
	By Lagrange's group theorem,
	$|\Gamma_{i+1}|$ divides $|\Gamma_i|$, and so either $\Gamma_{i+1}=\Gamma_i$
	or $|\Gamma_{i+1}|\leq\frac12|\Gamma_i|$.
	Hence the number of strict refinement steps in our chain of subgroups is $h\leq\log|\Gamma|$.
	Since $|\Gamma|\leq(2d)!^n$, we get $h=\mathcal{O}(nd\log d)$.
	Furthermore, the $d$-tuples of levels giving our $h$ strict refinement
	steps can be, one at each step, computed by Algorithm~\ref{alg:Gammaprime}.
	
	\begin{algorithm}[tb]
		\caption{One step of computation of the subgroup $\Gamma'\subseteq\Gamma$}
		\label{alg:Gammaprime}\normalsize
		\begin{algorithmic}[1]\smallskip
			\Require the attachment collection $\mathcal{U}$ of the graph $K$;
			
			\noindent the colored poset $R$ of $K$ 
			with levels $L_1,L_2,\ldots,L_k$;
			
			\noindent a subgroup $\Gamma_{i-1}$ (via a generator set) of the full
			automorphism group of~$R$.
			\Ensure either a certificate that $\mathcal{U}$ is Venn-good 
			for every member of $\Gamma_{i-1}$; or
			
			\noindent
			a subgroup $\Gamma_i\subsetneq\Gamma_{i-1}$ (via a generator set)
			such that, for the attachment collection $\mathcal{T}\subseteq\mathcal{U}$
			of some $d$-tuple of levels of $R$,~ $\mathcal{T}$ is
			Venn-good precisely for every member of $\Gamma_i$ 
			(and not for members of $\Gamma_{i-1}\setminus\Gamma_i$).
			\smallskip
			
			\State $M_2 \leftarrow \emptyset$
			\Repeat{~for $a:=1,2,\ldots,d\,$}:
			\Repeat{~for $b:=1,2,\ldots,k+1\,$}:
			\If{$b>k$}
			\Return ``$\mathcal{U}$ is Venn-good for all of $\Gamma_{i-1}$'';
			\label{it:allgood}\EndIf{}
			\State $M_1 \leftarrow $ $(L_1\cup L_2\cup\dots\cup L_b)$ $\cup$ $M_2$;
			\State $\,\mathcal{U}_1 \leftarrow $
			the subcollection of attachment sets of~$M_1$,
			~$\mathcal{U}_1\subseteq\mathcal{U}$;
			\Until{ $\mathcal{U}_1$ is not Venn-good (cf.~Lemma~\ref{lem:testconsist})
				for some generator of $\Gamma_{i-1}$};
			\State $j_{a} \leftarrow b$;
			\State $M_2 \leftarrow $ $L_{j_1}\cup L_{j_2}\cup\dots\cup L_{j_{a}}$;
			
			\Until{$j_a=1$ or $a=d$ };
			\State $\,\mathcal{U}_2 \leftarrow $
			the subcollection of attachment sets of~$M_2$,
			~$\mathcal{U}_2\subseteq\mathcal{U}$;
			
			\State Call the algorithm of Theorem~\ref{thm:furstgen} to compute
			the subgroup $\Gamma_i\subseteq\Gamma_{i-1}$, such that the
			membership test of $\varrho\in\Gamma_i$ checks whether
			$\mathcal{U}_2$ is Venn-good for $\varrho$
			(cf.~Lemma~\ref{lem:testconsist});
			\State\Return { $\Gamma_i$ }
			
		\end{algorithmic}
	\end{algorithm}

	\begin{lem}\label{lem:Gammaprim}
		Computation of the subgroup $\Gamma'\subseteq\Gamma$ can be accomplished
		in \textbf{FPT}-time with respect to the fixed parameter~$d$ by iterated calls to
		Algorithm~\ref{alg:Gammaprime}, starting from $\Gamma_0=\Gamma$.
	\end{lem}
	\vspace*{-1ex}%
	\begin{proof}
		We give this proof by analyzing a call to Algorithm~\ref{alg:Gammaprime}.
		If $\mathcal{U}$ is Venn-good for every generator of $\Gamma_{i-1}$
		(which is equivalent to that of every member of $\Gamma_{i-1}$), then we
		find this already in the first iteration of $a=1$, on line \ref{it:allgood}.
		Hence we may further assume that $\mathcal{U}$ is not Venn-good.
		
		Let $k\geq j'_1>j'_2>\dots>j'_c\geq1$ be an index sequence of length
		$c\leq d$ such that the subcollection of attachment sets of
		$L_{j'_1}\cup L_{j'_2}\cup\dots\cup L_{j'_c}$ is not
		Venn-good for some generator of $\Gamma_{i-1}$
		(such a sequence must exist by Corollary~\ref{cor:d-tuples}),
		and the vector $(j'_1,j'_2,\ldots,j'_c)$ is lexicographically minimal of
		these properties.
		Then one can straightforwardly verify that $(j'_1,j'_2,\ldots,j'_c)$ is a
		prefix of (or equal to) the vector $(j_1,j_2,\ldots,j_a)$ computed by Algorithm~\ref{alg:Gammaprime}.
		Consequently, the collection $\mathcal{U}_2$ of
		Algorithm~\ref{alg:Gammaprime} is not Venn-good for some generator of $\Gamma_{i-1}$.
		
		Next, we verify the fulfillment of the assumptions of Theorem~\ref{thm:furstgen}.
		Generators of $\Gamma_{i-1}=\Pi$ have been given to
		Algorithm~\ref{alg:Gammaprime}.
		The ratio $|\Pi|/|\Pi_1|$, where $\Pi_1=\Gamma_i$ in our case,
		can be bounded as follows (despite we do not know $\Gamma_i$ yet):
		$|\Gamma_{i-1}|/|\Gamma_{i}|$ equals the number of distinct cosets of the
		subgroup $\Gamma_{i}$ in $\Gamma_{i-1}$.
		If we consider two automorphisms $\alpha,\beta\in\Gamma_{i-1}$ which are
		equal when restricted to $M_2$, then the automorphism $\alpha^{-1}\beta$
		determines a permutation of $\mathcal{U}$ which is identical on
		$\mathcal{U}_2$ (so it is Venn-good), and hence $\alpha^{-1}\beta\in\Gamma_{i}$.
		The latter means that $\alpha$ and $\beta$ belong to the same coset of
		$\Gamma_{i}$, and consequently, the number of distinct cosets is at most the
		number of distinct subpermutations on $M_2$ possibly induced by
		$\Gamma_{i-1}$, that is at most $(2d)!^{d}$.
		Therefore, we can finish this step in \textbf{FPT}-time with respect to~$d$.
		
		Finally, as argued already, there can be at
		most $=\mathcal{O}(nd\log d)$ calls to Algorithm~\ref{alg:Gammaprime} altogether
		(and the last one certifies that $\mathcal{U}$ is already Venn-good).\qed
	\end{proof}

	\subsection{The complete algorithm}
	
	At last we review a detailed implementation of {\bf Procedure~\ref{proc:isomore}} as a pseudocode in Algorithm~\ref{alg:fullISO} listing all steps of the full isomorphism testing procedure.
	
	\begin{algorithm}[tb]
		\caption{Isomorphism test for general $S_d$-graphs ($G$,$H$)}
		\label{alg:fullISO}\normalsize
		\begin{algorithmic}[1]
			\Require Given two $S_d$-graphs $G$ and $H$ (their representations not required),
			and the parameter $d$ of $S_d$.
			\Ensure Result of the isomorphism test between $G$ and $H$.
			\medskip
			
			\State Find the maximal clique collections $\mathcal{S}$ and 
			$\mathcal{T}$ of the chordal graphs~$G$~and~$H$;
			\If{$|\mathcal{S}|$ $\neq$ $|\mathcal{T}|$, or the
				cardinalities of members of $\mathcal{S}$ and $\mathcal{T}$ do not match}
			\State\Return ``$G$ and $H$ are not isomorphic'';
			\EndIf{}
			
			\State $K \leftarrow G\uplus H$ (disjoint union);
			\Repeat{~for each maximal clique $C$ $\in$ $\mathcal{S}$}:
			\State Find the connected components $X_1,X_2,\ldots,X_k$ of~$G-C$,
			assuming the \mbox{equivalent} connected components
			are joined into single bridge(s) of~$C$;
			\State $P \leftarrow$ the partial order on the connected components
			$X_1,X_2,\ldots,X_k$ of~$G-C$ determined by their
			attachments;
			\Until{all $G[C \cup X_i]$, $1\leq i\leq k$, are interval graphs,
				and $P$ is of width $\leq d$};
			\label{li:fixCC}
			
			\For{each $D$ $\in$ $\mathcal{T}$ with $\vert C \vert = \vert D \vert$} 
			\State Find the connected components $Y_1,Y_2,\ldots,Y_l$ of $H-D$,
			assuming the \mbox{equivalent} connected components
			are joined into single bridge(s) of~$D$;
			
			\State $Q \leftarrow$ the partial order on the conn.\ components
			$Y_1,Y_2,\ldots,Y_l$ of $H-D$;
			\If{$k=l$, all $H[D \cup Y_i]$, $1\leq i\leq k$, are interval graphs,
				\\~\hfill\mbox{and $Q$ is of width $\leq d$}}
			
			\State $R \leftarrow P \uplus Q$ (disjoint union);
			\label{li:letR}
			\State Determine the levels $L_1,L_2,\ldots,L_k$ of the poset~$R$;
			\For{$i:=1,2,\ldots,k$}	
			\For{each pair of components $Z,Z'\in L_i$}
			\State Compare the interval graphs $K(Z)$ and $K(Z')$ to respectful isomorphism
			\label{it:comparZ}
			(recall the notation $K(Z)$ from Procedure~\ref{proc:isomore});
			\EndFor{}
			
			\State Having computed in the previous step the isomorphism equivalence classes 
			$L_i^1,\ldots,L_i^{c_i}$ of~$L_i$,
			give the element $Z\in L_i$ of poset $R$ color $(i,j)$ iff $Z\in L_i^j$;
			\EndFor{}
			
			\State $\Gamma \leftarrow$ the color-preserving automorphism group
			of~$R$, computed by Theorem~\ref{thm:bdcm};
			\label{li:cGamma}
			\State $\Gamma' \leftarrow$ the subgroup of $\Gamma$ consisting of
			those automorphisms $\varrho$ of $R$ which pass the test of Lemma~\ref{lem:preciseVenn};
			computed by {\bf Algorithm~\ref{alg:Gammaprime}};
			\label{li:cGammaX}
			
			\If{$P$ and $Q$ are swapped in some of the generators of $\Gamma'$}
			\State\Return ``$G$ and $H$ are isomorphic'';
			\label{li:retyes}
			\EndIf{}
			\EndIf{}
			\EndFor{}
			\State\Return ``$G$ and $H$ are not isomorphic'';
		\end{algorithmic}
	\end{algorithm}
	
	\begin{theorem}\label{thm:MAIN}
		The {\sc Isomorphism} problem of $S_d$-graphs can be solved in
		\textbf{FPT}-time with respect to the fixed parameter~$d$ by Algorithm~\ref{alg:fullISO}.
	\end{theorem}
	
	\begin{proof}
		We apply Procedure~\ref{proc:isomore}.
		This process runs $\mathcal{O}(n)$ iterations of choices of $C$ and~$D$;
		we first greedily find any valid central clique $C$ of an
		$S_d$-representation of $G$, and then iterate all maximal cliques~$D\subseteq H$.
		In each iteration, we routinely compute in polynomial time the sets of components
		$\mathcal{X}$ and $\mathcal{Y}$, and the central posets $P$ and $Q$ on them and $R=P\uplus Q$.
		If any of $P,Q$ has width greater than $d$, then we reject this iteration.
		Similarly, we reject the iteration if any of the graphs $K(Z)$ for
		$Z\in\mathcal{X}\cup\mathcal{Y}$ is not interval~\cite{recogIntervalLinear}.
		Otherwise, we compute colors on $R$ such that $Z,Z'\in\mathcal{X}\cup\mathcal{Y}$
		receive the same color iff they are on the same level of $R$ and 
		$K(Z)\simeq K(Z')$ with a respectful isomorphism.
		For the latter we use the isomorphism algorithm of~\cite{recogIntervalLinear} with coloring of~$C\cup D$.
		This finishes step 1.
		
		Step 2 -- computing the automorphism group $\Gamma$ of $R$, is done by Corollary~\ref{cor:bdcm}.
		
		Step 3 -- finding the subgroup $\Gamma'\subseteq\Gamma$,
		is accomplished by an iterated application of Theorem~\ref{thm:furstgen};
		the refinement steps are defined by $d$-tuples of levels of $R$ according to 
		Corollary~\ref{cor:d-tuples} and the above outlined procedure for finding them,
		and there are $\mathcal{O}(nd\log d)$ such steps.
		The membership test used in Theorem~\ref{thm:furstgen} is provided by
		Lemma~\ref{lem:testconsist}.

		Finally, we straightforwardly check in step~4 on the generators of $\Gamma'$ whether 
		some of them maps an element of $P$ to an element~of~$Q$.
		
		If any iteration of Procedure~\ref{proc:isomore} succeeds in step 4, then, by
		Lemma~\ref{lem:preciseVenn}, there exists an automorphism of the graph
		$K=G\uplus H$ which moreover swaps $G$ and~$H$.
		Then~$G\simeq H$.
		Conversely, assume $G\simeq H$.
		Since $G$ is an $S_d$-graph, we find a maximal central clique $C\subseteq G$,
		and since all maximal cliques $D\subseteq H$ are tried,
		we get into an iteration with $D$ being the isomorphic image of~$C$.
		Then the posets $P$ and $Q$ (wrt.~$C,D$) must be isomorphic respecting
		their colors, which follows from $G\simeq H$.
		Therefore, there exists an automorphism $\varrho\in\Gamma'$, and hence also
		a generator of $\Gamma'$, swapping $P$ and $Q$ (and preserving the
		computed cardinality Venn diagram by Lemma~\ref{lem:preciseVenn}).
		Some iteration hence succeeds and returns that $G\simeq H$.\qed
	\end{proof}

	\begin{restatable}{remark}{runtimMAIN}
		We do not explicitly state the runtime in Theorem~\ref{thm:MAIN}
		partly since it is not really useful and since neither \cite{furst}
		which we use states explicit runtime.
		Here we briefly remark that Procedure~\ref{proc:isomore} loops
		$\mathcal{O}(n)$ times with suitable $C$ and different choices of $D$, 
		analogously to the procedure of Theorem~\ref{theo:SDSMallTheo}, 
		and this initial setup of the procedure altogether takes time
		$\mathcal{O}(d^2n^3)$.
		Then we have to account for $\mathcal{O}(n)$ calls to steps 2 and 3 of
		Procedure~\ref{proc:isomore}, that is, $\mathcal{O}(n)$ computations of the
		subgroups $\Gamma$ and $\Gamma'$.
		Each time this part is dominated by step 3 which performs $\mathcal{O}(nd\log d)$
		calls to the algorithm of Theorem~\ref{thm:furstgen}~\cite{furst} in order
		to compute~$\Gamma'$ from~$\Gamma$.
		Reading the fine details of \cite{furst}, and adjusting it (the ``sift
		table'') to our setting, leads to an estimate of
		$\mathcal{O}(n^3)\cdot d!^{\mathcal{O}(d)}$ for each of these calls.
		After summarizing, we get the total estimate of
		$\mathcal{O}(n^5)\cdot d!^{\mathcal{O}(d)}$.
	\end{restatable}
	\begin{remark}
		Theorem~\ref{thm:MAIN} easily extends to colored graphs; 
	\begin{itemize}
		\item on line \ref{it:comparZ} of Algorithm~\ref{alg:fullISO} we can test for 
		colored isomorphism of interval	graphs as well, and
		\item in the calls to Algorithm~\ref{alg:Gammaprime} we
		simply test the cardinality Venn diagrams separately in each color class.
	\end{itemize}
	\end{remark}
	
This may also be a useful consequence.
	
	\begin{cor}\label{cor:Sdautomo}
		The automorphism group $Aut(G)$ of an $S_d$-graph $G$ can be computed in
		\textbf{FPT}-time with respect to the fixed parameter~$d$.
	\end{cor}
	
	Before we turn to the proof, we remark that the temptingly easy way to this
	corollary -- to take just the part of Algorithm~\ref{alg:fullISO} dealing
	with the graph $G$ and the automorphisms of the poset $P$ -- would give us
	only the subgroup of the automorphism group of $G$ set-wise stabilizing~$C$.
	This may not be complete $Aut(G)$.
	
	\begin{proof}
		We run slightly modified Algorithm~\ref{alg:fullISO} on two copies of the
		graph $G$ (i.e.,~$H=G$).
		On line \ref{li:cGammaX}, we denote the computed group by
		$\Gamma'_D=\Gamma'$ (since it depends on the choice of~$D$ in the loop).
		The actual modification comes on line \ref{li:retyes};
		over all choices $D$ such that $\Gamma'_D$ swaps $P$ and~$Q$,
		we record the following permutations $\varrho^*$ of $V(K)$ for every generator
		$\varrho$ of $\Gamma'_D$:
		\begin{itemize}
			\item[a)] The permutation $\varrho^*$ is such that for every $Z\in\mathcal{X}\cup\mathcal{Y}$
			we let $\varrho^*$ on $Z$ be any isomorphism from $K(Z)$ to $K(\varrho(Z))$
			(e.g., on line \ref{it:comparZ} of the algorithm).
			\item[b)] On $C\cup D$, we let the restriction of $\varrho^*$ be any subpermutation respecting the
			mapping $\tilde\varrho$ on the attachment collection of $\mathcal{X}\cup\mathcal{Y}$
			in $C\cup D$ (cf.~Lemma~\ref{lem:preciseVenn}).
			This is a correct definition; recall that the attachments of $Z$ and of
			$\varrho(Z)$ are invariant on the choice of isomorphism in (a).
		\end{itemize}
		
		Let $\Delta_1$ be the permutation group generated on $V(K)$ by the
		permutations $\varrho^*$ for all generators $\varrho$ of $\Gamma'_D$ and over
		all matching choices of~$D$.
		Let $\Delta_2$ be the group obtained by adding to $\Delta_1$ the generators of
		the automorphism groups of the interval graphs $G[C\cup X]$
		for~$X\in\mathcal{X}$, computed as in~\cite{recogIntervalLinear}.
		Let $\Delta_3$ be the group obtained by adding to $\Delta_2$ the generators
		of the symmetric groups on subsets of $C$ which are the cells of the
		attachment collection $\mathcal{U}$ of $\mathcal{X}$ in~$C$.
		Finally, let $\Delta$ be the restriction of $\Delta_3$ to $V(G)$.
		We claim $Aut(G)=\Delta$.
		
		Indeed, every member of $\Delta$ is an automorphism of $G$ which follows
		from the analysis of Algorithm~\ref{alg:fullISO}.
		On the other hand, any automorphism $g\in Aut(G)$ is ``discovered'' by
		Algorithm~\ref{alg:fullISO} as a member of $\Gamma'_D$ for $D=g(C)$, 
		modulo the concrete choices of isomorphisms between the interval components
		and of a permutation of $C\cup D$ which respects the attachment collection
		$\mathcal{U}$ of $\mathcal{X}$ in~$C$,
		and $g$ thus is contained in the constructed group~$\Delta$.
		\qed\end{proof}

	\section{An XP-time isomorphism approach for $T$-graphs}\label{hardestsection}
	
	We now show that the approach from Section~\ref{hardsection} can be extended
	to solve even the isomorphism problem of $T$-graphs.
	However, this relatively easy extension comes at a price -- the more general
	algorithm runs only in \textbf{XP}-time.

	Before getting to the extension, we note that there is another possible approach
	to $T$-graph isomorphism, to which we will return in the next section.
	Given two $T$-graphs $G$ and $H$ on $n$ vertices, we may consider all possible 
	assignments of maximal cliques of $G$ and $H$ to the branching nodes of~$T$.
	This can be achieved in $n^{\mathcal{O}(T)}$ iterations.
	In each iteration, we use the assignment to define $S_d$-graph isomorphism
	subproblems ``around'' each branching node of $T$, and apply the algorithm
	from Section~\ref{hardsection}.
	However, this seemingly straightforward procedure has some rather deep technical
	problems when considering general $T$-graphs,
	and that is why we prefer another recursive approach outlined next.

	Our actual approach combines the building blocks of Algorithm~\ref{alg:fullISO}
	with a special recursive routine which, in a nutshell, takes care of
	the non-interval pieces which occur as bridges of the chosen clique(s).
	It is described in the next procedure.
	
	\begin{proc}\label{proc:xpT}\rm~
		Given two $T$-graphs $G$ and $H$ on $n$ vertices, we compute the
		automorphism group $Aut(G\uplus H)$ if $G$ and $H$ are isomorphic, or output
		that $G$ and $H$ are not isomorphic.
		
		First, fix a maximal clique $C \subseteq G$ such that every component of
		$G-C$ is a $T'$-graph for some strict subtree $T'\subsetneq T$.
		Then apply the following procedure to each maximal clique $D \subseteq H$ 
		such that $\vert C \vert = \vert D \vert$:
		\begin{enumerate}	
			\item Identify the set $\mathcal{X}$ of all components $Z$ of $G-C$ such that $K(Z)$ is an interval graph,
			and the set $\mathcal{Y}$ of analogous components of $H-D$.
			If $\vert \mathcal{X} \vert = \vert \mathcal{Y} \vert$,
			then let $A_1, \dots, A_a$ and $B_1, \dots, B_b$ be the remaining (i.e., non-interval) components of $G-C$ and $H-D$, respectively.
			Note that the numbers $a,b$ of the latter components are bounded as $a,b\leq|V(T)|$.
			If $a=b$, proceed with the next steps.
			
			\item\label{it:st2b} As in Algorithm~\ref{alg:fullISO},
			construct the posets $P$ and $Q$ formed by the interval components
			$\mathcal{X}$ and $\mathcal{Y}$, respectively, color their nodes by respectful
			isomorphism of the components, and compute the color-preserving automorphism group
			$\Gamma_0$ of $R = P \uplus Q$.
			(Recall that the group $\Gamma_0$ may not yet preserve the cardinality Venn
			diagram of the attachment collection of $\mathcal{X}\cup\mathcal{Y}$ in $C\cup D$.)
			If $\Gamma_0$ contains a generator swapping $P$ and $Q$, proceed with the next steps.
			
			\item Loop through all bijections $m$ from $\mathcal{A}=\{ A_1, \dots, A_a \}$ to $\mathcal{B}=\{ B_1, \dots, B_a \}$:
			\begin{enumerate}
				\label{it:loopbm}%
				\item\label{it:loopbma} For each matching pair of subgraphs $G[C\cup A_i]$ and
				$H[D\cup B_{m(i)}]$, $i=1,\ldots,a$, compute recursively the automorphism group $$\Gamma_i :=
				Aut\big(\, G[C\cup A_i] \uplus H[D\cup B_{m(i)}] \,\big) $$
				set-wise stabilizing~$C\cup D$.
				If the graphs $G[C\cup A_i]$ and $H[D\cup B_{m(i)}]$ are not isomorphic this way, 
				stop this iteration and continue with the next bijection~$m$.

				\item Define the group $\Gamma$ as the direct product of the following groups
				\label{it:compstogether}
				\begin{itemize}\parskip-1pt\par
					\item of $\Gamma_0$ on $\mathcal{X}\cup\mathcal{Y}$ (step~\ref{it:st2b}), and
					\item for $i=1,\ldots,a$, of $\Gamma_i$ restricted to~$A_i\cup B_{m(i)}$ (step~\ref{it:loopbma}).
				\end{itemize}
				See Remark~\ref{rem:explainaut} for an explanation of the purpose of this step.
				
				\item Let $\mathcal{U}$ be the attachment collection of
				$\mathcal{X}\cup\mathcal{Y}$ in $C\cup D$, and $\mathcal{V}_i$ be the
				attachment collection of $A_i\cup B_{m(i)}$ in $C\cup D$.
				Let $\mathcal{V}=\mathcal{U}\cup \mathcal{V}_1\cup\ldots\cup\mathcal{V}_a$ (recall that this is a multiset).
				Define $\tilde\Gamma$ as the group of permutations of $\mathcal{V}$ generated
				by the action of $\Gamma$ on the attachment collection $\mathcal{V}$.
				\label{it:gaction}

				\item\label{it:usealg2} Using analogy of Algorithm~\ref{alg:Gammaprime}
				 (see details in the proof of Theorem~\ref{thm:xpT}),
				compute the subgroup $\Gamma'\subseteq\Gamma$ of those members of $\Gamma$
				which are consistent on $C\cup D$.
				Precisely, $\Gamma'$ is the maximum subgroup of $\Gamma$ such that the related
				group $\tilde\Gamma'$, defined as in \eqref{it:gaction}, preserves the cardinality Venn diagram of $\mathcal{V}$.
				If no member of $\Gamma'$ swaps $G-C$ with $H-D$, stop this iteration and go to the next~$m$.
				\\Otherwise, let $\Gamma'$ be denoted by $\Gamma'_m$ w.r.t.\ the current bijection~$m$.
				
			\end{enumerate}		
			
			\item\label{it:st4} If none of the bijections $m$ in step \ref{it:loopbm} successfully computed the group $\Gamma'_m$,
			then the result of the current choice of the clique $D\subseteq H$ is void.
			Otherwise, let $\Gamma'_D$ be the group generated by all $\Gamma'_m$ over such ``successful''~$m$
			for the current clique~$D$.
			This is a sound definition of $\Gamma'_D$ since every member of each $\Gamma'_m$ is actually a permutation of the set
			$(\mathcal{X}\cup\mathcal{Y})\cup\bigcup_{i=1,\ldots,a}(A_i\cup B_{m(i)})$ (cf.~step~\ref{it:compstogether}).

			\item\label{it:st5} At last, we use the same construction on the components of $\mathcal{X}\cup\mathcal{Y}$ as in
			Corollary~\ref{cor:Sdautomo} to turn $\Gamma'_D$ into an automorphism subgroup
			(not complete yet) acting on the vertex set $V(G\uplus H)$.
			Let this group be denoted by $\Delta_D$.
		\end{enumerate}
		
		At the end, the desired automorphism group $Aut(G\uplus H)$ is generated by
		the union of the generators of the non-void groups $\Delta_D$ computed in
		the procedure.
	\end{proc}
	
	\begin{remark}\label{rem:explainaut}
	To fully understand Procedure~\ref{proc:xpT}, it is important to notice why the group $\Gamma$ in
	step~\eqref{it:compstogether} is such seemingly over-complicated;
	the group acts on the interval components (as whole) of~$\mathcal{X}\cup\mathcal{Y}$, while it
	simultaneously acts on the individual vertices of the non-interval components.

	The reason for such handling is that, as we know from Section~\ref{hardsection}, the
	attachments of a single component from~$\mathcal{X}\cup\mathcal{Y}$ in $C\cup D$ form a chain by the inclusion
	(and so they are invariant on automorphism of the component), but the attachments of a single non-interval
	component from~$\{A_i,B_i:i=1,\ldots,a\}$ in $C\cup D$ may be ``shuffled'' by an automorphism of the
	component (and hence we need to ``track'' the mapping of individual vertices of that component throughout the algorithm).
	On the other hand, detailed tracking of the attachments of the non-interval
	components can be done efficiently since the number of them is bounded from above by~$|V(T)|$.
	\end{remark}	

	\begin{remark}
		Procedure~\ref{proc:xpT} implicitly tests for representability as $T$- and
		$T'$-graphs; this can be done using the \textbf{XP}-time algorithm of~\cite{zemanWG}.
		On the other hand, we do not explicitly use intersection representations of
		the graphs, and we only need two properties implied by $T$-representability:
		that the number of non-interval components and the recursion depth is
		bounded in~$|V(T)|$, and that there can be only bounded number of pairwise incomparable
		attachments in the collection~$\mathcal{V}$.
		We may thus just assume $T$- or $T'$-representability of our graphs without checking, and
		throw away the computation branches in which we detect violation of the implied properties.
	\end{remark}

	\begin{theorem}\label{thm:xpT}
		The isomorphism problem of $T$-graphs can be solved in
		\textbf{XP}-time with respect to the fixed parameter~$|V(T)|$ by Procedure~\ref{proc:xpT}.
	\end{theorem}
	
	\begin{proof}
		The run of Procedure~\ref{proc:xpT}, without the recursive calls from it,
		takes \textbf{FPT}-time with the size of $T$ as the parameter.
		This follows analogously to the proof of Theorem~\ref{thm:MAIN}; in particular,
		since the maximum number of incomparable neighborhoods (attachment sets) of all vertices
		of $(G-C)\uplus(H-D)$ in $C\cup D$ is clearly in $\mathcal{O}(|V(T)|)$ for $T$-graphs.
		Though, the depth of the recursion may be of order up to $\Theta(|V(T)|)$;
		every recursive call from Procedure~\ref{proc:xpT} assumes $T'$-graphs for
		some $T'$ strictly smaller than original~$T$.
		Hence the overall \textbf{XP}-time of order $n^{\mathcal{O}(|V(T)|)}$.
		(We remark, without a proof, that the exponent can be improved to
		$\mathcal{O}(diam(T)+\log|V(T)|)$ by a suitable choice of $C$ and
		management of the recursion, but such small improvement is not worth the effort,
		as we explain in the conclusion section.)
		
		It can be easily seen that, as in Section~\ref{hardsection}, every generator
		$\gamma$ of the computed groups $\Delta_D$ in Procedure~\ref{proc:xpT} really is an
		automorphism of the graph~$G\uplus H$.
		Indeed, $\gamma$ consists of a permutation of the set
		$(\mathcal{X}\cup\mathcal{Y})\cup\bigcup_{i=1,\ldots,a}(A_i\cup B_{m(i)})$
		(step~\ref{it:st4} of Procedure~\ref{proc:xpT}) for some bijection $m$ from step~\ref{it:loopbm},
		and a composition of respectful automorphisms
		of the individual components of~$\mathcal{X}\cup\mathcal{Y}$ (step~\ref{it:st5}).
		Moreover, since $\gamma\in\Gamma'_m$, the restriction of $\gamma$ to $G[C\cup A_i] \uplus H[D\cup B_{m(i)}]$
		is an automorphism of that subgraph set-wise stabilizing $C\cup D$, for each~$i=1,\ldots,a$.
		And since the group $\Gamma'_m$ consists only of permutations which are consistent on~$C\cup D$,
		$\gamma$ preserves also the edges of $G\uplus H$ which have one end on $C\cup D$ 
		and the other end in any component of~$(G-C)\uplus(H-D)$.

		Henceforth, if there is a generator of $\Delta_D$ swapping $G$ and $H$, we correctly
		conclude that $G$ and $H$ are isomorphic.
		
		Conversely, if $G$ and $H$ are two isomorphic $T$-graphs, then there
		exists a maximal clique $C\subseteq G$, as assumed by Procedure~\ref{proc:xpT}, and 
		the isomorphic image $D\subseteq H$ of~$C$.
		The interval components of $G-C$ are mapped by isomorphism to the interval components of $H-D$,
		which defines an automorphism of the poset~$R$ (cf.~$\Gamma_0$).
		The non-interval components of $G-C$ are mapped bijectively to the non-interval
		components of $H-D$, defining the bijection~$m$ (step \ref{it:loopbm}).
		The recursive automorphism groups on these non-interval components matched
		by $m$ are computed correctly by induction.
		Since these particular automorphisms in the procedure come from an actual
		automorphism of $G\uplus H$ swapping $G$ and $H$, they are consistent on~$C\cup D$
		and the original automorphism thus is a member of the computed group $\Gamma'_m$.
		
		Lastly, we justify step~\ref{it:usealg2}; that Algorithm~\ref{alg:Gammaprime}
		can be used in the setting of Procedure~\ref{proc:xpT}.
		By the proof of Lemma~\ref{lem:Gammaprim} it suffices to distribute the vertices of $A_i\cup B_{m(i)}$ adjacent to $C\cup D$,
		for each~$i=1,\ldots,a$, into bounded-size ``levels'' (clustered by exactly same neighborhood) which are invariant upon respectful automorphisms.
		For instance, we can easily define such levels by equal cardinalities of the neighborhoods of $A_i\cup B_{m(i)}$ in $C\cup D$,
		and this finishes the proof.
		\qed\end{proof}
	
	Since the algorithm of Theorem~\ref{thm:xpT} does not care about the shape
	of the tree $T$ (and degree-$2$ nodes of $T$ are irrelevant), we actually
	get the following more general conclusion:
	
	\begin{corollary}
		The isomorphism problem of chordal graphs of leafage at most $\ell$ can be solved in
		\textbf{XP}-time with respect to the fixed parameter~$\ell$.
	\end{corollary}

	\section{Isomorphism of proper $T$-graphs in FPT-time}\label{proper}
	
	As noted in the introduction, proper $T$-graphs present a significantly more
	restrictive graph class than general $T$-graphs.
	Recently, Chaplick et al.~\cite{chaplick2020recognizing} have shown that proper $T$-graphs
	can be recognized in \textbf{FPT}-time parameterized by the size of $T$
	(which is not yet known for general $T$-graphs).
	Here we show that also the isomorphism problem of proper $T$-graphs is
	notably easier than that of $S_d$-graphs and $T$-graphs in general.
	Though, the isomorphism problem stays GI-complete for proper $S_d$-graphs
	with $d$ on the input -- recall Proposition~\ref{prop:SdGIc},
	and so we cannot realistically hope for polynomial-time algorithms here.
	
	In the case of $T=S_d$, the algorithm is really simple and purely combinatorial:
	
	\begin{theorem}\label{properSD}
		The isomorphism of proper $S_d$-graphs can be tested in \textbf{FPT}-time parameterized by~$d$.
	\end{theorem}
	
	\begin{proof}
		This is based on the following observation:
		
		If $C \subseteq G$ is the central clique of a proper $S_d$-representation of~$G$
		and $X_1,\dots,X_k$ are the connected components of $G-C$ which have nonempty attachment in~$C$,
		then~$k\leq d$.
		Indeed, if $k>d$, then some two components, say $X_i,X_j$, would be
		represented on the same ray of a subdivision $S_d'$ of~$S_d$ such that
		$X_i$ is closer to the center than~$X_j$.
		If $X_j$ is adjacent to $v\in C$, then the representation of $v$ in $S_d'$
		actually contains the whole representation of $X_i$,
		which is a contradiction to $G$ being a proper $S_d$-graph.
		
		Assume we get two proper $S_d$-graphs $G$ and~$H$ on $n$ vertices.
		Analogously to the previous algorithms, we hence fix a central maximal clique
		$C \subseteq G$, determine $X_1,\dots,X_k$ which are the connected components of
		$G-C$ which have nonempty attachment in~$C$, and denote by $X_0$ the
		possible connected components of $G$ disjoint from~$C$ (or set $X_0=\emptyset$).
		We loop through all maximal cliques $D \subseteq H$ such that $|D|=|C|$,
		and similarly determine components $Y_1,\ldots,Y_l$ and $Y_0$ of $H-D$.
		If $k=l$, we compare $X_0$ and $Y_0$ to (proper) interval graph isomorphism~\cite{recogIntervalLinear}.
		Then we loop through all bijections $\varrho:\{1,\ldots,k\}\to\{1,\ldots,k\}$,
		compare the graphs $G[C \cup X_i]$ and $H[D \cup Y_{\varrho(i)}]$,
		$i=1,\ldots,k$, to interval graph isomorphism,
		and compare their attachment collections in $C$ and in $D$ using the test
		provided by Lemmas~\ref{lem:preciseVenn} and~\ref{lem:testconsist}.
		(Recall that the attachment collections are automorphism-invariant.)
		If we ever succeed, then we mark $G$ and $H$ as isomorphic.
		
		The overall approach takes $\mathcal{O}(k!n^3)\leq\mathcal{O}(d!\,n^3)$ time which is in \textbf{FPT} with respect to $d$.  
		The correctness follows by the same arguments as in Section~\ref{hardsection}.  
		\qed\end{proof}

	We now move onto proper $T$-graphs. 
	As with $S_d$-graphs, we may always restrict to $T$-representations of graphs $G$ (in a subdivision $T'$ of~$T$)
	such that the cliques represented at the branching nodes of $T'$ are maximal cliques of $G$,
	and we call such cliques the {\em branching cliques} of this representation.
	Informally, our idea is to use the finding of Chaplick et al.~\cite{chaplick2020recognizing}
	that those maximal cliques of a proper $T$-graph~$G$ which occur as the branching cliques
	of any $T$-representation of $G$ are somehow special, and that their number, modulo an easy canonical adjustment, is bounded.
	We will call such cliques (after the adjustment) the {\em rich cliques} of~$G$ here.
	Then it will not be difficult to try all mappings between the rich cliques in \textbf{FPT}-time,
	as mentioned already in Section~\ref{hardestsection}, and then answer the isomorphism problem
	similarly as in the proof of Theorem~\ref{properSD}.

	The full definition of the terms, following \cite{chaplick2020recognizing}, is quite technical
	and we skip the details here since they are not important for our proof.
	For our purpose, it is enough to accept the notion of a {\em chain} in a proper $T$-graph~$G$,
	which is (without further unnecessary details that can be found in~\cite{chaplick2020recognizing})
	a collection $\mathcal{Y}$ of maximal cliques of $G$ such that $\mathcal{Y}$ is partitioned 
	into a sequence $(C_1,\ldots,C_s)$ of the {\em inner} cliques and into two {\em terminals}
	which are nonempty subcollections $\mathcal{T}_1,\mathcal{T}_2\subset\mathcal{Y}$ of (remaining) cliques.
	Then we use the following:

	\begin{prop}[Chaplick et al.~\cite{chaplick2020recognizing}]\label{prop:chains}	
		Assume $G$ is a connected proper $T$-graph.
		Then the following sets are unique in $G$ (meaning that they appear the same in every proper $T$-representation of $G$);
		the set $\mathcal{H}(G)$ of the chains in $G$, and
		the set $\mathcal{S}(G)$ of the remaining maximal cliques of $G$ not contained in the chains of $\mathcal{H}(G)$.
		Furthermore:
		\begin{itemize}
			\item[a)] The inner cliques of every chain of $\mathcal{H}(G)$ form a path in any proper $T$-representation of $G$,
			\item[b)] if a terminal $\mathcal{T}_1$ of a chain in $\mathcal{H}(G)$ is not a singleton clique,
			then $\mathcal{T}_1$ is the set of the inner cliques of some (other) chain in $\mathcal{H}(G)$, and
			\item[c)] if an inner clique $C_i$ of a chain is a branching clique in some proper $T$-repre\-sentation of $G$,
			then $C_i$ is contained in a terminal $\mathcal{T}_1$ of some (other) chain in $\mathcal{H}(G)$, and
			any (other) clique $C_j\in\mathcal{T}_1$ may replace $C_i$ in a proper $T$-representation of $G$.
		\end{itemize}

		The sets $\mathcal{H}(G)$ and $\mathcal{S}(G)$ can be computed in \textbf{FPT}-time parameterized by $|T|$,
		and their cardinalities are in~$\mathcal{O}(|T|^2)$.

	\end{prop}

	Note that, even though the number of chains (and hence also the number of terminals) is bounded with respect to $T$ in Proposition~\ref{prop:chains},
	the number of potential branching cliques is not bounded since the terminals may have arbitrary cardinality.
	However, Proposition~\ref{prop:chains}(b) allows us to select, from a non-singleton terminal $\mathcal{T}_1$,
	only the extreme two (the first and the last) of the sequence of inner cliques.

	We hence determine the collection $\mathcal{R}$ of {\em rich cliques} of a connected proper $T$-graph $G$ algorithmically as follows:
	\begin{enumerate}
		\item Using the Algorithm of \cite{chaplick2020recognizing}, identify the set $\mathcal{H}(G)$ of the chains in~$G$,
		and $\mathcal{S}(G)$ of remaining maximal cliques of $G$ in \textbf{FPT}-time parameterized by~$\vert T \vert$.
		\item Initially set $\mathcal{R}:=\mathcal{S}(G)$.
		\item For each terminal of every chain in $\mathcal{H}(G)$ which is a single clique, add~it~to~$\mathcal{R}$.
		\item For every pair of chains in $\mathcal{H}(G)$; if a non-singleton terminal of one is the set of the inner cliques
		of the other chain, ordered as the sequence $(C_1,\ldots,C_s)$, then add $C_1$ and $C_s$ to~$\mathcal{R}$.
	\end{enumerate}

	\begin{lem}\label{rich2}
		If $G$ is a connected proper $T$-graph, then the set of rich cliques $\mathcal{R}$ of $G$ is determined as above
		in \textbf{FPT}-time parameterized by $\vert T \vert$, its cardinality is in~$\mathcal{O}(|T|^2)$, and
		$\mathcal{R}$ is isomorphism-invariant.
	\end{lem}

	\begin{proof}
		The runtime is in \textbf{FPT} only because of the Algorithm of \cite{chaplick2020recognizing}, while the rest
		of the computation of $\mathcal{R}$ is clearly in polynomial time.
		The cardinalities $|\mathcal{H}(G)|=\mathcal{O}(|T|^2)$ and $|\mathcal{S}(G)|=\mathcal{O}(|T|^2)$ are by Proposition~\ref{prop:chains},
		and for every chain in $\mathcal{H}(G)$, at most $4$ cliques are added to $\mathcal{R}$.
		Finally, Proposition~\ref{prop:chains} also guarantees that the set $\mathcal{R}$ is invariant on the choice of
		a particular proper $T$-representation of $G$, hence isomorphism-invariant.
	\qed\end{proof}
	
	We can now give our fully combinatorial algorithm:

	\begin{theorem}\label{properT}
		The isomorphism of proper $T$-graphs can be tested in \textbf{FPT}-time with respect to the size of~$T$.
	\end{theorem}
	
	\begin{proof}
		We are given two proper $T$-graphs $G$ and $H$ on $n$ vertices.
		Let $V_b(T)$ denote the set of branching nodes of~$T$ (which will be the set of branching nodes of any subdivision of $T$, too).
		For simplicity, we assume
		\begin{itemize}
			\item 
			that both $G$ and $H$ are connected, as otherwise we can test the isomorphism between each pair of components separately, and
			\item
			that $G$ and $H$ are not proper $T_1$-graphs for any $T_1\subsetneq T$,
			since we can exhaustively try the coming algorithm for all $T_1\subsetneq T$ before trying with~$T$.
		\end{itemize}
		
		Imagine now that $G\simeq H$. Then there exist proper $T$-representations of $G$ and $H$ in the same subdivision $T'$ of $T$,
		such that the branching cliques of $G$ are mapped to the branching cliques of $H$ in a chosen isomorphism $f$ of $G$ and $H$.
		Let these branching cliques be $C_1,\ldots,C_m\subseteq G$ and $D_1,\ldots,D_m\subseteq H$, where~$m=|V_b(T)|$ and~$f(C_i)=D_i$.
		By Proposition~\ref{prop:chains}, we may also assume that these cliques are among the rich cliques of $G$ and of~$H$.
		We define the graphs $G_0:=G-(C_1\cup\ldots\cup C_m)$ and $H_0:=G-(D_1\cup\ldots\cup D_m)$.
		Then, since $C_1,\ldots,C_m$ are the branching cliques in an intersection representation of $G$ in $T'$,
		we have that every connected component $G_1\subseteq G_0$ is adjacent to at most two cliques ($C_i,C_j$) among $C_1,\ldots,C_m$,
		and we denote by $G_1^+:=G_1\cup C_i\cup C_j$ which is always an interval graph. We do the same for~$H_0$.

		On the other hand, every branching clique $C_i$ is adjacent to at most $\ell$ components of $G_0$, where $\ell$ is the number of leaves of~$T$.
		If this was not true, then there would be two intervals in $T'$ (indeed intervals since all branching cliques have been subtracted from~$G_0$)
		representing two vertices of distinct components adjacent to $C_i$, and so one lying on a path in $T'$ from the branching node of $C_i$ to the other.
		That would violate that our representation is proper.
		In particular, the number of connected components of $G_0$ and of $H_0$ is in $\mathcal{O}(|T|^2)$.

		Since $f$ is an isomorphism from $G$ to~$H$, there is a pairing of the components of $G_0$ and $H_0$ such that
		for $G_1\subseteq G_0$ and $H_1\subseteq H_0$, we have $f(G_1)=H_1$, and this can be verified as an interval graph isomorphism.
		The attachment collections of the components in the branching cliques can then be compared (between $G$ and $H$)
		in polynomial time the same way as in the proof of Theorem~\ref{properSD}.
		Altogether, we can verify an isomorphism from $G$ to $H$ in polynomial time once we ``guess'' the right matching selections
		of branching cliques in $G$ and~$H$.

		Our algorithm hence first identifies the sets of all rich cliques $\mathcal{C}$ and $\mathcal{D}$
		of $G$ and $H$, respectively, using Lemma~\ref{rich2}.
		Then we loop through all assignments $f:V_b(T)\to\mathcal{C}$ and $g:V_b(T)\to\mathcal{D}$ of rich cliques to the branching nodes,
		and for each assignment we check whether the graph $G_1^+$ (resp., $H_1^+$) over all components of $G_0$ (resp., of~$H_0$) as above is an interval graph.
		If this test succeeds, then we check whether there is an isomorphism from $G$ to $H$ respecting the assignment $f$ and $g$, as described above.
		If no isomorphism is found in any of the rounds, then we conclude that~$G\not\simeq H$.

		As for the correctness, if the isomorphism test ever succeeds, then obviously $G\simeq H$.
		If $G\simeq H$, then, as argued above, there exist ``matching'' proper $T$-representations of $G$ and $H$ with their rich cliques
		as the branching cliques, this assignment of the rich cliques is among the tested ones, and hence the outcome will be that $G\simeq H$.
	\qed\end{proof}

	\section{Conclusions}
	
	In this paper, we have focused on the graph isomorphism problem on $S_d$-graphs and $T$-graphs. 
	We have shown that $S_d$-graph isomorphism problem includes testing for the
	isomorphism of posets of bounded width which can be solved in
	\textbf{FPT}-time using a classical group-based approach by Furst, Hopcroft and
	Luks~\cite{furst} via Babai~\cite{babai-bdcm}. 
	We have given an \textbf{FPT}-time algorithm to test the isomorphism of
	$S_d$-graphs with the parameter~$d$,
	which also builds on the mentioned classical group-based approach.  
	
	Since it does not seem easy to solve the isomorphism problem of posets
	of bounded width in a combinatorial (or algebra-free) way, and we
	are not aware of any published result in this direction, the related question 
	of an existence of a purely combinatorial \textbf{FPT}-time algorithm for $S_d$-graph isomorphism
	remains open (due to Theorem~\ref{theo:red2}).
	We have only shown a partial answer to this open question in the case of more 
	restricted proper $S_d$-graphs and $T$-graphs in Section~\ref{proper}.
	
	As a natural extension of our \textbf{FPT}-time algorithm for $S_d$-graph isomorphism,
	we have shown that the isomorphism problem for $T$-graphs can be solved in \textbf{XP}-time with respect to the size of~$T$.
	In the direct approach we have taken here, the jump from \textbf{FPT}- to \textbf{XP}-time
	seems unavoidable. However, in subsequent very recent papers independently
	Arvind, Nedela, Ponomarenko and Zeman~\cite{DBLP:journals/corr/abs-2107-10689}
	and these authors~\cite{DBLP:conf/walcom/CagiriciH22} have found different ways to
	address the isomorphism problem for $T$-graphs in \textbf{FPT}-time.
	The approach of~\cite{DBLP:conf/walcom/CagiriciH22} builds on the
	core ideas of Section~\ref{hardsection}, combined with a tricky
	canonical decomposition of $T$-graphs, while
	\cite{DBLP:journals/corr/abs-2107-10689} use a deeply algebraic approach.
	
	Chaplick et al.\ also showed that the graph isomorphism problem for $H$-graphs is
	GI-complete when $H$ contains a double triangle as a minor \cite{zemanWG}. 
	However, it is an open problem what is the complexity of the $H$-graph
	isomorphism problem when $H$ contains one cycle, and we would like to consider this problem in a future research.

	\begin{acknowledgements}
		We would like to thank to Pascal Schweitzer for pointing us to the paper~\cite{furst},
		and to Onur \c{C}a\u{g}{\i}r{\i}c{\i} for comments on this manuscript.
	\end{acknowledgements}

	\bibliography{Sd-bibliography}

\begin{thebibliography}{10}

\bibitem{AHU}
Alfred~V. Aho, John~E. Hopcroft, and Jeffrey~D. Ullman.
\newblock {\em The Design and Analysis of Computer Algorithms}.
\newblock Addison-Wesley, 1974.

\bibitem{DBLP:journals/corr/abs-2107-10689}
Vikraman Arvind, Roman Nedela, Ilia Ponomarenko, and Peter Zeman.
\newblock Testing isomorphism of chordal graphs of bounded leafage is
  fixed-parameter tractable.
\newblock {\em CoRR}, abs/2107.10689, 2021.
\newblock \href {http://arxiv.org/abs/2107.10689} {\path{arXiv:2107.10689}}.

\bibitem{babai-bdcm}
L{\'{a}}szl{\'{o}} Babai.
\newblock Monte {C}arlo algorithms in graph isomorphism testing.
\newblock {\em Tech.~Rep.~79-10, Universit\'e de Montr\'eal}, 1979.
\newblock 42 pages.

\bibitem{DBLP:conf/stoc/Babai16}
L{\'{a}}szl{\'{o}} Babai.
\newblock Graph isomorphism in quasipolynomial time [extended abstract].
\newblock In Daniel Wichs and Yishay Mansour, editors, {\em Proceedings of the
  48th Annual {ACM} {SIGACT} Symposium on Theory of Computing, {STOC} 2016,
  Cambridge, MA, USA, June 18-21, 2016}, pages 684--697. {ACM}, 2016.
\newblock \href {https://doi.org/10.1145/2897518.2897542}
  {\path{doi:10.1145/2897518.2897542}}.

\bibitem{biro}
Miklós Biró, Mihály Hujter, and Zsolt Tuza.
\newblock Precoloring extension. i. interval graphs.
\newblock {\em Discrete Mathematics \textbf{100}}, pages 267--279, 1992.

\bibitem{recogIntervalLinear}
Kellogg~S. Booth and George~S. Lueker.
\newblock Testing for the consecutive ones property, interval graphs, and graph
  planarity using {PQ}-tree algorithms.
\newblock {\em J. Comput. Syst. Sci.}, 13(3):335--379, 1976.
\newblock \href {https://doi.org/10.1016/S0022-0000(76)80045-1}
  {\path{doi:10.1016/S0022-0000(76)80045-1}}.

\bibitem{treedepth}
Adam Bouland, Anuj Dawar, and Eryk Kopczy\'{n}ski.
\newblock On tractable parameterizations of graph isomorphism.
\newblock In Dimitrios~M. Thilikos and Gerhard~J. Woeginger, editors, {\em
  Parameterized and Exact Computation - 7th International Symposium, {IPEC}
  2012, Ljubljana, Slovenia, September 12-14, 2012. Proceedings}, volume 7535
  of {\em Lecture Notes in Computer Science}, pages 218--230. Springer, 2012.
\newblock \href {https://doi.org/10.1007/978-3-642-33293-7_21}
  {\path{doi:10.1007/978-3-642-33293-7_21}}.

\bibitem{DBLP:conf/walcom/CagiriciH22}
Deniz~A\u{g}ao\u{g}lu \c{C}a\u{g}{\i}r{\i}c{\i} and Petr Hlin\v{e}n{\'{y}}.
\newblock Isomorphism testing for t-graphs in {FPT}.
\newblock In {\em {WALCOM}}, volume 13174 of {\em Lecture Notes in Computer
  Science}, pages 239--250. Springer, 2022.
\newblock arXiv:2111.10910.
\newblock \href {https://doi.org/10.1007/978-3-030-96731-4_20}
  {\path{doi:10.1007/978-3-030-96731-4_20}}.

\bibitem{chaplick2020recognizing}
Steven Chaplick, Petr~A. Golovach, Tim~A. Hartmann, and Dusan Knop.
\newblock Recognizing proper tree-graphs.
\newblock In Yixin Cao and Marcin Pilipczuk, editors, {\em 15th International
  Symposium on Parameterized and Exact Computation, {IPEC} 2020, December
  14-18, 2020, Hong Kong, China (Virtual Conference)}, volume 180 of {\em
  LIPIcs}, pages 8:1--8:15. Schloss Dagstuhl - Leibniz-Zentrum f{\"{u}}r
  Informatik, 2020.
\newblock \href {https://doi.org/10.4230/LIPIcs.IPEC.2020.8}
  {\path{doi:10.4230/LIPIcs.IPEC.2020.8}}.

\bibitem{zemanWG}
Steven Chaplick, Martin T\"{o}pfer, Jan Voborn{\'{\i}}k, and Peter Zeman.
\newblock On {H}-topological intersection graphs.
\newblock In Hans~L. Bodlaender and Gerhard~J. Woeginger, editors, {\em
  Graph-Theoretic Concepts in Computer Science - 43rd International Workshop,
  {WG} 2017, Eindhoven, The Netherlands, June 21-23, 2017, Revised Selected
  Papers}, volume 10520 of {\em Lecture Notes in Computer Science}, pages
  167--179. Springer, 2017.
\newblock \href {https://doi.org/10.1007/978-3-319-68705-6_13}
  {\path{doi:10.1007/978-3-319-68705-6_13}}.

\bibitem{zeman2}
Steven Chaplick and Peter Zeman.
\newblock Combinatorial problems on {H}-graphs.
\newblock {\em Electron. Notes Discret. Math.}, 61:223--229, 2017.
\newblock \href {https://doi.org/10.1016/j.endm.2017.06.042}
  {\path{doi:10.1016/j.endm.2017.06.042}}.

\bibitem{isoSplitGIComp}
Fan R.~K. Chung.
\newblock On the cutwidth and the topological bandwidth of a tree.
\newblock {\em SIAM J. Alg. Discr. Meth.}, 6:268--277, 1985.

\bibitem{DBLP:journals/networks/Colbourn81}
Charles~J. Colbourn.
\newblock On testing isomorphism of permutation graphs.
\newblock {\em Networks}, 11(1):13--21, 1981.
\newblock \href {https://doi.org/10.1002/net.3230110103}
  {\path{doi:10.1002/net.3230110103}}.

\bibitem{eigenmul}
Sergei Evdokimov and Ilia~N. Ponomarenko.
\newblock Isomorphism of coloured graphs with slowly increasing multiplicity of
  jordan blocks.
\newblock {\em Comb.}, 19(3):321--333, 1999.
\newblock \href {https://doi.org/10.1007/s004930050059}
  {\path{doi:10.1007/s004930050059}}.

\bibitem{posetbook}
Peter~A. Fejer and Dan~A. Simovici.
\newblock Partially ordered sets.
\newblock {\em In: Mathematical Foundations of Computer Science. Texts and
  Monographs in Computer Science. Springer, New York, NY}, pages 127--175,
  1991.

\bibitem{DBLP:conf/esa/FominGR18}
Fedor~V. Fomin, Petr~A. Golovach, and Jean{-}Florent Raymond.
\newblock On the tractability of optimization problems on {H}-graphs.
\newblock In Yossi Azar, Hannah Bast, and Grzegorz Herman, editors, {\em 26th
  Annual European Symposium on Algorithms, {ESA} 2018, August 20-22, 2018,
  Helsinki, Finland}, volume 112 of {\em LIPIcs}, pages 30:1--30:14. Schloss
  Dagstuhl - Leibniz-Zentrum f{\"{u}}r Informatik, 2018.
\newblock \href {https://doi.org/10.4230/LIPIcs.ESA.2018.30}
  {\path{doi:10.4230/LIPIcs.ESA.2018.30}}.

\bibitem{furst}
Merrick~L. Furst, John~E. Hopcroft, and Eugene~M. Luks.
\newblock Polynomial-time algorithms for permutation groups.
\newblock In {\em 21st Annual Symposium on Foundations of Computer Science,
  Syracuse, New York, USA, 13-15 October 1980}, pages 36--41. {IEEE} Computer
  Society, 1980.
\newblock \href {https://doi.org/10.1109/SFCS.1980.34}
  {\path{doi:10.1109/SFCS.1980.34}}.

\bibitem{furst-bdcm}
Merrick~L. Furst, John~E. Hopcroft, and Eugene~M. Luks.
\newblock A subexponential algorithm for trivalent graph isomorphism.
\newblock In {\em Proc. 11th Southeastern Conf. Combinatorics, Graph Theory,
  and Computing, Congressum Numerantium 3}, 1980.

\bibitem{chordalityInters}
Fǎnicǎ Gavril.
\newblock The intersection graphs of subtrees in trees are exactly the chordal
  graphs.
\newblock {\em Journal of Combinatorial Theory, Series B}, 16(1):47--56, 1974.
\newblock \href {https://doi.org/10.1016/0095-8956(74)90094-X}
  {\path{doi:10.1016/0095-8956(74)90094-X}}.

\bibitem{planarLinear}
John~E. Hopcroft and J.~K. Wong.
\newblock Linear time algorithm for isomorphism of planar graphs (preliminary
  report).
\newblock In Robert~L. Constable, Robert~W. Ritchie, Jack~W. Carlyle, and
  Michael~A. Harrison, editors, {\em Proceedings of the 6th Annual {ACM}
  Symposium on Theory of Computing, April 30 - May 2, 1974, Seattle,
  Washington, {USA}}, pages 172--184. {ACM}, 1974.
\newblock \href {https://doi.org/10.1145/800119.803896}
  {\path{doi:10.1145/800119.803896}}.

\bibitem{DBLP:journals/corr/Kawarabayashi15a}
Ken{-}ichi Kawarabayashi.
\newblock Graph isomorphism for bounded genus graphs in linear time.
\newblock {\em CoRR}, abs/1511.02460, 2015.
\newblock \href {http://arxiv.org/abs/1511.02460} {\path{arXiv:1511.02460}}.

\bibitem{KLAVIK201585}
Pavel Klav{\'{\i}}k, Jan Kratochv{\'{\i}}l, Yota Otachi, and Toshiki Saitoh.
\newblock Extending partial representations of subclasses of chordal graphs.
\newblock {\em Theor. Comput. Sci.}, 576:85--101, 2015.
\newblock \href {https://doi.org/10.1016/j.tcs.2015.02.007}
  {\path{doi:10.1016/j.tcs.2015.02.007}}.

\bibitem{treewidth}
Daniel Lokshtanov, Marcin Pilipczuk, Michal Pilipczuk, and Saket Saurabh.
\newblock Fixed-parameter tractable canonization and isomorphism test for
  graphs of bounded treewidth.
\newblock {\em {SIAM} J. Comput.}, 46(1):161--189, 2017.
\newblock \href {https://doi.org/10.1137/140999980}
  {\path{doi:10.1137/140999980}}.

\bibitem{DBLP:journals/jcss/Luks82}
Eugene~M. Luks.
\newblock Isomorphism of graphs of bounded valence can be tested in polynomial
  time.
\newblock {\em J. Comput. Syst. Sci.}, 25(1):42--65, 1982.
\newblock \href {https://doi.org/10.1016/0022-0000(82)90009-5}
  {\path{doi:10.1016/0022-0000(82)90009-5}}.

\bibitem{mckee1999topics}
Terry~A. McKee and Fred~R. McMorris.
\newblock {\em Topics in Intersection Graph Theory}.
\newblock Discrete Mathematics and Applications. Society for Industrial and
  Applied Mathematics (SIAM), 1999.

\bibitem{genus}
Gary~L. Miller.
\newblock Isomorphism testing for graphs of bounded genus.
\newblock In Raymond~E. Miller, Seymour Ginsburg, Walter~A. Burkhard, and
  Richard~J. Lipton, editors, {\em Proceedings of the 12th Annual {ACM}
  Symposium on Theory of Computing, April 28-30, 1980, Los Angeles, California,
  {USA}}, pages 225--235. {ACM}, 1980.
\newblock \href {https://doi.org/10.1145/800141.804670}
  {\path{doi:10.1145/800141.804670}}.

\bibitem{DBLP:conf/esa/Neuen21}
Daniel Neuen.
\newblock Isomorphism testing parameterized by genus and beyond.
\newblock In Petra Mutzel, Rasmus Pagh, and Grzegorz Herman, editors, {\em 29th
  Annual European Symposium on Algorithms, {ESA} 2021, September 6-8, 2021,
  Lisbon, Portugal (Virtual Conference)}, volume 204 of {\em LIPIcs}, pages
  72:1--72:18. Schloss Dagstuhl - Leibniz-Zentrum f{\"{u}}r Informatik, 2021.
\newblock \href {https://doi.org/10.4230/LIPIcs.ESA.2021.72}
  {\path{doi:10.4230/LIPIcs.ESA.2021.72}}.

\bibitem{recogChordaLinear}
Donald~J. Rose, Robert~Endre Tarjan, and George~S. Lueker.
\newblock Algorithmic aspects of vertex elimination on graphs.
\newblock {\em {SIAM} J. Comput.}, 5(2):266--283, 1976.
\newblock \href {https://doi.org/10.1137/0205021} {\path{doi:10.1137/0205021}}.

\bibitem{isoPoset}
Ryuhei Uehara, Seinosuke Toda, and Takayuki Nagoya.
\newblock Graph isomorphism completeness for chordal bipartite graphs and
  strongly chordal graphs.
\newblock {\em Discret. Appl. Math.}, 145(3):479--482, 2005.
\newblock \href {https://doi.org/10.1016/j.dam.2004.06.008}
  {\path{doi:10.1016/j.dam.2004.06.008}}.

\bibitem{isoChordalGIComp}
Viktor~N. Zemlyachenko, Nickolay~M. Korneenko, and Regina~I. Tyshkevich.
\newblock Graph isomorphism problem.
\newblock {\em J. of Soviet Mathematics}, 29:1426--1481, 1985.

\end{thebibliography}

\end{document}